\documentclass[11pt]{article}
\usepackage[a4paper,margin=1in]{geometry}
\usepackage{amsmath,amsthm,amssymb}
\usepackage{graphicx}
\usepackage{tikz}
\usepackage{todonotes}

\newtheorem{theorem}{Theorem}
\newtheorem{corollary}{Corollary}
\newtheorem{lemma}{Lemma}

\newtheorem{conjecture}{Conjecture}

\theoremstyle{definition}
\newtheorem{definition}{Definition}
\newtheorem{claim}{Claim}

\theoremstyle{remark}
\newtheorem*{remark}{Remark}
\newtheorem*{notation}{Notation}

\newcommand{\R}{\mathbb{R}}

\newcounter{note}[section]

\title{On the Nisan-Ronen conjecture for submodular valuations}

\author{George Christodoulou\thanks{Department of Computer Science,
    University of Liverpool, UK.  Email:
    \texttt{gchristo@liverpool.ac.uk} }
  \and Elias Koutsoupias\thanks{Department of Computer Science, University
    of Oxford, UK. Email: \texttt{elias@cs.ox.ac.uk}}
  \and Annam{\'a}ria
  Kov{\'a}cs\thanks{Department of Informatics, Goethe University,
    Frankfurt M., Germany. Email: \texttt{panni@cs.uni-frankfurt.de}}
  }

\date{}
\begin{document}
\maketitle

\begin{abstract}
  We consider incentive compatible mechanisms for a domain that is
  very close to the domain of scheduling $n$ unrelated machines: the
  single exception is that the valuation of just one machine is
  submodular. For the scheduling problem with such cost functions, we
  give a lower bound of $\Omega(\sqrt{n})$ on the approximation ratio
  of incentive compatible deterministic mechanisms. This is a strong
  information-theoretic impossibility result on the approximation
  ratio of mechanisms on relatively simple domains. The lower bound of
  the current work assumes no restriction on the mechanism side, but
  an expanded class of valuations, in contrast to previous general
  results on the Nisan-Ronen conjecture that hold for only special
  classes of mechanisms such as local, strongly monotone, and
  anonymous mechanisms. Our approach is based on a novel
  characterization of appropriately selected smaller instances that
  allows us to focus on particular type of algorithms (linear
  mechanisms), from which we extract a locality property that gives
  the lower bound.
\end{abstract}

\section{Introduction}

The design of protocols that provide appropriate incentives to
participants, entice them to cooperate, and behave in a way which is
socially beneficial has a long and celebrated history. It is the realm
of mechanism design which is one of the most researched branches of
Game Theory and Microeconomics.  It studies the design of algorithms,
called {\em mechanisms}, and it has numerous applications to many
situations in modern societies, whenever a protocol of conduct of
selfish participants is required. The mechanism asks each participant
to bid their preferences over the different social outcomes, and
implements one of them (e.g. the one which is most socially
beneficial). The challenge is that the preferences of the participants
are {\em private}, and they are either unmotivated to report
correctly, or strongly motivated to report them erroneously, if a
false report is profitable. A truthful mechanism provides incentives
in a way that it is in the best interest of each participant to bid
{\em truthfully}.

The algorithmic nature of mechanism design and the associated
computational issues, were brought to light, in the seminal
20-year-old paper by Nisan and Ronen~\cite{NR01}, which essentially
established the area of algorithmic mechanism design.
They proposed the scheduling problem on unrelated machines, a
fundamental, extensively studied from the algorithmic perspective,
optimization problem as a representative specimen to study the
limitations of truthful mechanisms. The objective is to incentivize
$n$ machines to execute $m$ tasks, so that the maximum completion time
of the machines, i.e. the makespan, is minimized.

Nisan and Ronen applied the famous Vickrey-Clarke-Groves (VCG)
mechanism~\cite{Vic61,Cla71,Gro73} which is a general machinery that
truthfully computes the outcome that maximizes the {\em social
  welfare}, which for the case of scheduling is the allocation that
minimizes the sum of completion times. The VCG is truthful and
polynomial-time for scheduling, but with respect to the makespan
minimization it has a rather poor approximation ratio, equal to the
number of machines $n$. Despite this, they conjectured that the VCG is
the mechanism with the best approximation ratio for this problem.

\begin{conjecture}
  There is no deterministic truthful mechanism with approximation
  ratio better than $n$ for the problem of scheduling $n$ unrelated
  machines.
\end{conjecture}

The bound in the conjecture is information-theoretic, in the sense
that it should hold for all mechanisms (algorithms), polynomial-time
or not. The Nisan-Ronen conjecture developed into one of the central
problems in Algorithmic Game Theory, and despite intensive efforts
only very sparse progress has been made towards its resolution. The
original Nisan-Ronen paper showed that no truthful deterministic
mechanism can achieve an approximation ratio better than $2$. This was
improved to $2.41$~\cite{ChrKouVid09}, and later to $2.61$
\cite{KV07}, which leaves a huge gap with the known upper bound of
$n$. The most general and interesting result, by Ashlagi et
al.\cite{ADL09}, resolved a restricted version of the conjecture, for
the special yet natural class of {\em anonymous} mechanisms. However,
the original conjecture (for non-anonymous mechanisms) is widely open.

The original conjecture was posed for the case of additive valuations,
where the total cost of each machine equals the sum of its individual
costs of each task. However, many mechanism design settings consider
more general valuations such as submodular, subadditive etc. In these
valuations, the cost of a machine that takes a set of tasks $S$ is a
function of $S$ that satisfies some natural properties. For example,
for subadditive valuations the cost of a bundle of tasks $S$ can be
any value bounded above by the sum of the cost of the individual tasks
in $S$. This extended domain may affect significantly the incentive
compatibility aspect of mechanisms because it allows a machine to also
lie about the cost of bundles of tasks.

It is natural to pose the Nisan-Ronen question for these extended
valuation classes. Despite the importance of the problem and various
attempts, the conjecture has been widely open for all of them. In this
work, we consider \emph{submodular cost functions}, where the marginal
contribution of a task to the total cost of a machine is a
non-increasing function. The class of submodular valuations contains
all additive valuations, it is a proper subject of subadditive
valuations, and it is one of the most restrictive natural valuations
for the scheduling question. Furthermore it is the most-studied class
of valuations for the corresponding maximization problem of
combinatorial auctions (see for
example~\cite{LehmannLN06,DobzinskiV16,Dobzinski16}).

\subsection{Our result}

We give the first non-constant lower bound that works for {\em all}
deterministic truthful mechanisms when the cost functions are
submodular.

\begin{theorem}
\label{thm:main}
There is no deterministic truthful mechanism with approximation ratio
better than $\sqrt{n-1}$ for the problem of scheduling $n$ unrelated
machines with submodular cost functions.
\end{theorem}

Actually we show a stronger result: Let $\mathcal{M}$ be the class of
scheduling mechanisms which are deterministic and truthful {\em when
  all machines are additive except for one machine, which is
  submodular}. Then no mechanism in $\mathcal{M}$ has an approximation
ratio better than $\sqrt{n-1}$ on the set of instances of scheduling
unrelated additive machines.

The lower bound of the current work assumes \emph{no restriction on
  the mechanism side}, but an expanded class of valuations. In
contrast, the
Nisan-Ronen conjecture has been shown to hold for only \emph{special classes of
mechanisms} (local~\cite{NR01}, strongly
  monotone~\cite{MualemS18}, and anonymous~\cite{ADL09} mechanisms).

\paragraph{Other valuation classes:}

Actually we don't anticipate any difficulty in translating, in fact
carbon-copying, our results to supermodular valuations; we will
include the details in the final version of the paper. Interestingly,
the class of additive valuations of the original Nisan-Ronen
conjecture is exactly the intersection of submodular and supermodular
valuations.

In fact, we provide a stronger version of our theorem by considering
submodular functions which are also $\epsilon$-additive, in the sense
that the execution time of a set of tasks is within an arbitrarily
small $\epsilon$ from the sum of the execution times of its tasks (see
Section~\ref{sec:preliminaries} for a precise definition). An
important feature of our lower bound construction is that the cost
functions of all machines are additive \emph{with the single exception
  of one machine which is submodular} for particular disjoint pairs of
tasks.

\subsection{Overview of the techniques}

We provide an overview of our approach for the lower bound.
We consider instances in which every task has a fixed large value,
practically infinite, for all except for two machines; one of the two
machines is always the submodular player (player $0$). We can assume
that these tasks are allocated to one of the two machines, otherwise
the approximation ratio is sufficiently high. This restriction of the
allocation to only two players per task, have been previously used
(e.g.~\cite{ChrKouVid09,KV07}). The main difference with previous
approaches is that we use properties of mechanisms that involve at
least three players, the submodular player and two other
players. Obtaining such multiplayer statements is the bottleneck for a
complete characterization of mechanisms in multiplayer domains.

\paragraph{Two-player characterization.} We first focus on the tasks
that can be allocated to a particular additive player and fix the
values of the remaining tasks. For each additive player, there are
only two such tasks, and the situation is very similar to two-player
and two-task special case in which the valuation of one player is
submodular. A core part of our proof, which is also of
independent interest, is a characterization of the allocation
functions of all truthful mechanisms for this case.

We provide a complete characterization for the case of two players
with submodular valuations which are non-negative and are bounded
above by a constant. This is essential in our construction to
guarantee that the fixed large value of the other machines does not
play any role. Since this is a minimization problem, it is the lower
bound (i.e., the restriction that the values are non-negative) that
creates complications rather then the upper bound on the domain. We
also provide a characterization for additive valuations, the actual
scheduling domain.

We note that similar two-player characterizations have been provided
by previous work \cite{DS08,ChristodoulouKV08}, for auctions and
scheduling domains, but none of these can be used in our approach. In
particular, the characterization of \cite{DS08} for scheduling relies
extensively on the bounded approximation for two players, but we need
a characterization without this assumption. The reason is that the
approximation ratio of any two players is in general unrelated to the
approximation ratio of the whole multi-player instance.  In the same
work, a characterization for auctions with subadditive valuations is
provided, but this is also of no use to the minimization we consider
here.
Finally, the characterization of~\cite{ChristodoulouKV08} cannot be
used because it allows negative values.

Indeed, as we show in Section~\ref{sec:characterization}, the
scheduling domain admits truthful mechanisms not present in the
previous characterizations of \cite{DS08,ChristodoulouKV08}, which we
call \emph{relaxed affine minimizers}. Such mechanisms are essentially
affine minimizers for large values, but they can have {\em non-linear}
boundaries for small values (see
Definition~\ref{dfn:relaxed-affine-minimizers} and
Figure~\ref{fig:ex-relaxed-minimizer}).

\paragraph{Gluing two-player mechanisms to multi-player
  linear mechanisms.}
From the characterization of two-payer two-task mechanisms and using
the fact that we are interested in mechanisms with small approximation
ratio for the whole multi-player instances, we are able to exclude all
mechanisms except of those that have affine boundaries
(Lemma~\ref{lem:linear}), that is, affine minimizers
{\em whose coefficients may depend on the values of the other tasks}.

One of the main technical steps is that we use the truthfulness of the
submodular player to show that the scaling coefficient of these linear
mechanisms do not actually depend on the values of the other tasks
(Lemma~\ref{lem:linearity}). The rest of the proof, which includes
the most complicated technical steps of this work, is to analyze the
properties of truthful linear algorithms that facilitate the proof of
the lower bound (Lemma~\ref{lem:critical}).

\subsection{Related work}
\label{sec:related-work}

The problem of scheduling unrelated machines is a typical
multi-dimensional mechanism design problem. In multi-dimensional
mechanism design, the valuation of each player for different outcomes
is expressed via a vector (one value for every outcome). In the case
of unrelated scheduling, this vector expresses the processing times of
a machine for each subset of tasks and can be succinctly represented
by an $m$-valued vector, one value for each task. 

An interesting special case, which is well-understood, is the
single-dimensional mechanism design in which the values of the vector
are linear expressions of a single parameter. The principal
representative is the problem of scheduling {\em related} machines,
where the cost of each machine can be expressed via a single
parameter, its {\em speed}. This was first studied by Archer and
Tardos~\cite{AT01} who showed that, in contrast to the unrelated
machines version, an algorithm that minimizes the makespan can be
truthfully implemented --- albeit in exponential time. It was
subsequently shown that truthfulness has essentially no impact on the
computational complexity of the problem. Specifically, a randomized
truthful-in-expectation\footnote{This is one of the two main
  definitions of truthfulness for randomized mechanisms, where
  truth-telling maximizes the expected utility of each player.} PTAS
was given in~\cite{DDDR11} and a deterministic PTAS was given
in~\cite{CK13}; a PTAS is the best possible algorithm even for the
pure algorithmic problem (unless $P=NP$).

The main obstacle in resolving the Nisan-Ronen conjecture is the lack
of clear algorithmic understanding of truthfulness for many players in
multi-dimensional domains. In contrast, we understand better
truthfulness for a single player. Saks and Yu~\cite{SY05} gave a nice,
complete characterization of deterministic truthful mechanisms for
convex domains which was later extended to truthful-in-expectation
randomized mechanisms in \cite{AK08}. This characterization states
that the class of allocations of truthful mechanisms is the class of
{\em weakly monotone} algorithms~\cite{BCR+06}. This is an elegant
characterization, but it has not been proved very useful for
mechanisms of many players, because it is difficult to combine
monotonicity of each individual player into a single global
condition. Its direct applications have provided only constant bounds
for makespan minimization~\cite{ChrKouVid09,MualemS18,KV07,NR01}. What
would be more useful is a characterization similar to the one provided
by the seminal work of Roberts \cite{Rob79} for {\em unrestricted}
domains. It essentially states that the only truthful mechanisms are affine
extensions of VCG. Similar characterizations have been provided in
\cite{ChristodoulouKV08, DS08,DobzinskiN15} for settings with only two
players. Extending these characterizations to multiple players for
scheduling and combinatorial auctions is notoriously hard, mainly due
to lack of externalities in these settings: the valuation of a player
for an allocation depends only on the subset of tasks it receives and
is indifferent on how the remaining of the tasks are assigned to the
other players\footnote{For two players, there exist implicit
  externalities as the tasks one player doesn't get determine what the
  other player gets.}. 

Scheduling is related to combinatorial auctions, where multiple items
need to be assigned to a set of buyers. This is a broad and successful
area, and the setting shares both aforementioned features of
multi-dimensionality and lack of externalities, therefore insights and
techniques can be transferred from the one problem to the
other. However the difference is that the objective for combinatorial
auctions is {\em social welfare maximization}, and this is known to be
achieved by the VCG mechanism, albeit in exponential time. Hence the
focus on this rich area is on what can be achieved by computationally
efficient mechanisms (see for example~\cite{DobzinskiV16}). But in the
case of the scheduling with the min-max objective, the flavor is more
information theoretic, as we know that {\em not even exponential time}
mechanisms can achieve the optimal makespan.

\subsubsection{Further related work}

Lavi and Swamy~\cite{LaviS09} proposed an interesting approach to
attack the Nisan-Ronen question, by restricting the input domain, but
still keep the multi-dimensional flavour of the setting. They assumed
that each entry in the input matrix can take only two possible values
``low'' and ``high'', that are publicly known to the designer. In this
case, they showed an elegant deterministic mechanism with an
approximation factor of 2. Surprisingly, even for this special case
there is a lower bound of $11/10$. Yu~\cite{Yu09} extended the results
for a range of values, and Auletta et al.~~\cite{Auletta0P15} studied
multi-dimensional domains where the private information of the
machines is a single bit.

Randomization has been explored and slightly improved the known
guarantees. There are two notions of truthfulness for randomized
mechanisms. A mechanism is {\em universally truthful} if it is defined
as a probability distribution over deterministic truthful mechanisms,
while it is {\em truthful-in-expectation}, if in expectation no player
can benefit by lying. In \cite{NR01}, a universally truthful mechanism
was proposed for the case of two machines, and was later extended to
the case of $n$ machines by Mu'alem and Schapira~\cite{MualemS18} with
an approximation guarantee of $0.875n$, which was later improved to
$0.837n$ by \cite{LuYu08}. Lu and Yu~\cite{LuY08a} showed a
truthful-in-expectation mechanism with an approximation guarantee of
$(m+5)/2$.
Mu'alem and Schapira~\cite{MualemS18}, showed a lower bound of $2-1/m$,
for both notions of randomization. Christodoulou, Koutsoupias and
Kov{\'a}cs~\cite{CKK10} extended the lower bound for fractional
mechanisms, where each task can be fractionally allocated to multiple
machines. They also showed a fractional mechanism with a guarantee of
$(m+1)/2$.
A sequence of papers studied randomized mechanisms for the special case
of two machines \cite{Lu09, LuY08a} where a tight answer on the
approximation factor is still unresolved. Currently, the best upper
bound is $1.587$ due to Chen, Du, and Zuluaga~\cite{ChenDZ15}.

The truthful implementation of other objectives have been explored by
Mu'alem and Schapi\-ra~\cite{MualemS18} for multi-dimensional problems
and by Epstein and van Stee~\cite{EpsteinLS13} for single-dimensional
ones, giving a PTAS for a wide range of objective functions. Leucci,
Mamageishvili and Penna~\cite{LeucciMP18} showed high lower bounds for
other min-max objectives on some combinatorial optimization
problems. In the Bayesian setting, Daskalakis and Weinberg
\cite{DaskalakisW15} showed a mechanism that is at most a factor of 2
from the {\em optimal truthful mechanism}, but not with respect to
optimal makespan. Chawla et al.~\cite{ChawlaHMS13} provided bounds of
prior-independent mechanisms (where the input comes from a probability
unknown to the mechanism). Giannakopoulos and
Kyropoulou~\cite{GiannakopoulosK17} showed that the VCG mechanism
achieves an approximation ratio of $O( \log n/\log \log n )$ under
some distributional and symmetry assumptions.

\section{Preliminaries}
\label{sec:preliminaries}

There is a set $N$ of $n$ machines and a set $M$ of $m$ tasks that
need to be scheduled on the machines. The processing time or cost that
each machine $i$ takes to process a subset $S$ of tasks is described
by a set function $t_i: 2^m \rightarrow \R_{\geq 0}$. In classic
unrelated machines scheduling, the cost functions are additive and the
objective is to minimize the makespan (min-max objective).

We also consider more general cost functions that are normalized
$(t_i(\emptyset)=0)$ and monotone ($t_i(S)\leq t_i(T)$ for $S\subseteq
T$). We focus on {\em submodular} cost functions, which satisfy the
following condition for every $S,T\subseteq M$ $$t_i(S\cup
T)+t_i(S\cap T)\leq t_i(S)+t_i(T).$$

In the classical Nisan-Ronen model, the cost functions are
additive. In our lower bound construction, we will assume that all
cost functions are additive except for one which is submodular. We
also consider valuations which are arbitrarily close to additive,
which we call $\epsilon$-additive, such that for every subset $S$,
$\sum_{j\in S}(t_i(\{j\}))-\epsilon \leq t_i(S)\leq \sum_{j\in
  S}(t_i(\{j\}))+\epsilon.$ Our results hold for valuations that are
both submodular and $\epsilon$-additive---even for valuations within
an $\epsilon$ multiplicative factor from additive.

\paragraph{Mechanism design setting.}
We assume that each machine $i\in N$ is controlled by a selfish agent
that is reluctant to process the tasks and the cost function $t_i$ is
private information known only to her (also called the {\em type} of
agent $i$). In the most general version of the problem, the set
$\mathcal{T}_i$ of possible types of agent $i$ consists of all vectors
$b_i\in \mathbb{R}_+^{2^m}.$ Let also $\mathcal{T} = \times_{i\in
  N}\mathcal{T}_i$ be the space of type profiles.

A mechanism defines for each player $i$ a set $\mathcal{B}_i$ of
available strategies, the player (agent) can choose from. We will
consider \emph{direct revelation} mechanisms, i.e.,
$\mathcal{B}_i=\mathcal{T}_i$ for all $i,$ meaning that the players strategies
are to simply report their types to the mechanism. A player may report
a false cost function $b_i\neq t_i$, if this serves her interests.

A mechanism $(A,P)$ consists of two parts:
\begin{description}
\item[An allocation algorithm:] The allocation algorithm $A$ allocates the jobs to the machines depending
  on the players' bids $b=(b_1,\ldots ,b_n)$. Let $\mathcal{A}$ be the
  set of all possible partitions of $m$ tasks to $n$ machines. The
  allocation function $A:\mathcal{T}\rightarrow \mathcal{A}$
  partitions the tasks into the $n$ machines; we denote by $A_i(b)$
  the subset of tasks assigned to machine $i$ in the bid profile $b$.

\item[A payment scheme:] The payment scheme $P=(P_1,\ldots,P_n)$ determines the payments also
  depending on the bid values $b.$ The functions $P_1,\ldots,P_n$ stand
  for the payments that the mechanism hands to each agent
  i.e. $P_i:\mathcal{T}\rightarrow \R$.
\end{description}

The {\em utility} $u_i$ of a player $i$ is the payment that he gets
minus the {\em actual} time that he needs to process the set of tasks
assigned to her, $u_i(b)=P_i(b)-t_i(A_i(b))$. We are interested in
\emph{truthful} mechanisms. A mechanism is truthful, if for every
player, reporting his true type is a \emph{dominant
  strategy}. Formally,

$$u_i(t_i,b_{-i})\geq u_i(t'_i,b_{-i}),\qquad \forall i\in [n],\;\; t_i,t'_i\in \mathcal{T}_i, \;\; b_{-i}\in \mathcal{T}_{-i},$$
where $\mathcal{T}_{-i}$ denotes the possible types of all players disregarding
$i.$

We use as an objective to evaluate the performance of a mechanism's
allocation algorithm the {\em makespan}, that is the maximum
completion time of a machine.  The makespan of the allocation
algorithm $A$ with respect to a given input $t$ is
$$Mech(t)\stackrel{\mathrm{def}}{=}\max_{i\in N}t_{i}(A_i(t)).$$
We aim at minimizing the makespan, hence the optimum
is $$Opt(t)=\min_{A\in \mathcal{A}}\max_{i\in N}t_i(A_i).$$ We
are interested in the approximation ratio of the mechanism's
allocation algorithm.  A mechanism $M$ is \emph{$c$-approximate}, if
the allocation algorithm is $c$-approximate, that is, if
$c\geq\frac{Mech(t)}{Opt(t)}\;$ for all possible inputs $t.$

We are looking for truthful mechanisms with low approximation ratio irrespective of the running time to compute $A$ and $P.$ In other words, our lower bounds do not make use of any computational assumptions.

\paragraph{Weak monotonicity.} 

A useful characterization of truthful mechanisms in terms of the
following monotonicity condition, helps us to get rid of the payments
and focus on the properties of the allocation algorithm.

\begin{definition}
An allocation algorithm $A$ is called {\em
  weakly monotone (WMON)} if it satisfies the following property: for
every two inputs $t=(t_i,t_{-i})$ and $t'=(t'_i,t_{-i})$, the
associated allocations $A$ and $A'$ satisfy $$t_i(A_i)-t_i(A'_i)\leq
t'_i(A_i)-t'_i(A'_i).$$

\end{definition}

It is well known that the allocation function of every truthful
mechanism is WMON~\cite{BCR+06}, and also that this is a sufficient
condition for truthfulness in convex domains~\cite{SY05}.

A useful tool in our proof relies on the following immediate
consequence of WMON, which holds in additive domains as well as in the other
domains that we consider. Intuitively, it states that when you fix the
values of all players for a subset of tasks (focus on a cut of your
domain), then the {\em restriction} of the allocation to the rest of
the tasks must remain weakly monotone.

\begin{lemma}\label{lem:restriction} Let $A$ be a WMON allocation. Let us fix an $(S,T)$  partition of $M$, and consider  only valuations $t_i$ of player $i$ that are additive across $S$ and $T$, i.e., for every
  $X\subseteq M$, $t_i(X)=t_i(X \cap S)+ t_i(X \cap T)$. Then the
  restriction of the allocation $A$ on $S$ is weakly monotone for each
   valuation fixed on the subsets of $T.$ 
\end{lemma}

\begin{proof}

$A$ is weakly monotone, therefore 
$$t_i(A_i)-t'_i(A_i)\leq t_i(A'_i)-t'_i(A'_i),$$
and additivity across $S$ and $T$ implies
$$t_i(A_i\cap S)-t'_i(A_i \cap S)\leq t_i(A'_i\cap S)-t'_i(A'_i\cap S),$$
as the values of subsets of $T$ are fixed, $t_i(A_i\cap
T)=t'_i(A_i\cap T)$ and $t_i({A'}_i\cap T)=t'_i({A'}_i\cap
T)$.

\end{proof}

The following lemma was essentially shown in \cite{NR01} and has been
a useful tool to show lower bounds for truthful mechanisms for several
variants (see for example \cite{ChrKouVid09,MualemS18,ADL09}). Although
this holds more generally, we only state it (and use it in
Section~\ref{sec:lb}) for additive valuations.

\begin{lemma}\label{lemma:tool}
  Let $t$ be a bid vector  of additive valuations, and let $S=A_i(t)$ be the subset assigned to
  player $i$ by the mechanism. For any bid vector $t'=(t'_i,t_{-i})$
  such that only the bid of machine $i$ has changed and in such a way
  that for every task in $S$ it has decreased (i.e., $t'_i(\{j\})<
  t_i(\{j\}), j\in S$) and for every other task it has increased
  (i.e., $t'_i(\{j\})> t_i(\{j\}), j\in M\setminus S$).  Then the mechanism does not
  change the allocation to machine $i$, i.e., $A_i(t')=A_i(t)=S$.
\end{lemma}

The main challenge in multi-player settings is that the allocation of
the other machines may change and the above condition makes no promise
about how this can happen.

\section{Characterization for two players}
\label{sec:characterization}

A core element of our lower bound proof is a characterization of
truthful mechanisms for two tasks and two players (called $t$-player,
and $s$-player). We first provide a characterization for additive
valuations $t=(t_1,t_2)$, and $s=(s_1,s_2)$, so that both $s_1$ and
$s_2$ are bounded by an arbitrarily large but fixed value $B.$ Then we
extend it when the $t$ player has submodular valuations
$t=(t_1,t_2,t_{12})$ (Theorem~\ref{theo:epschar}), which is the main
element that we need in the lower bound in Section~\ref{sec:lb}. We
note that both the characterization and the lower bound result hold
analogously when the $t$-player has arbitrary monotone, or
$\epsilon$-additive (or submodular \emph{and} $\epsilon$-additive)
valuations. 

We postpone the details of the characterization proof to the Appendix,
but in this section we introduce the definitions that we will need in
order to state the main result and also to introduce some notation
which will be used in our lower bound proof in Section~\ref{sec:lb}.
We keep in the Appendix intact the whole proof of the
characterization.

\subsection{Basic Definitions}

Let $(A,P)$ be a truthful mechanism, where $A$ is the WMON allocation
function, and $P$ denotes the payment function. For input $(t,s)$ the
allocation is $A(t,s).$ Since we have only tasks 1 and 2, in $A(t,s)$
we can denote the allocation to one of the players as
$\alpha_t,\alpha_s\in\{12,1,2,\emptyset\}.$ 

For given $s\in[0,B)\times [0,B)$ the allocation for the $t$-player as
function of his \emph{own} bids $(t_1,t_2)$ (or $(t_1,t_2,t_{12})$) is denoted by $A[s],$ and
symmetrically $A[t]$ is an allocation function for the $s$-player.
For $\alpha_t\in\{12,1,2,\emptyset\},$ the allocation \emph{regions}
$R_{\alpha_t}(s)\subseteq \mathbb R_{\geq 0}^2$ (resp. $\mathbb R_{\geq 0}^3$) of $A[s]$ are defined
to be the \emph{interior} (wrt. $\mathbb R_{\geq 0}^2$) of the set of
all $t$ values such that $A(t,s)=\alpha_t.$ For the $s$-player we denote the respective regions by $R_{\alpha_s}^s(t).$

\begin{figure}[t]
\centerline{\includegraphics[height=4.5cm]{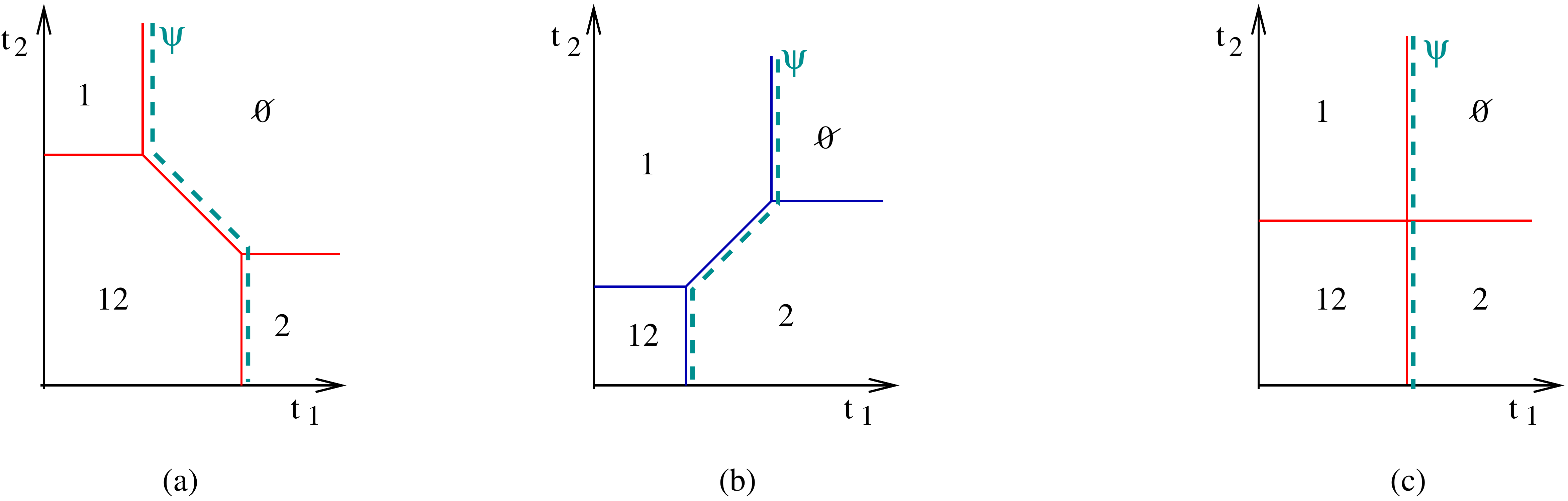}}
\caption{The allocation $A[s]$ to the $t$-player depending on his own bid vector $(t_1,t_2):\quad$ \newline  (a) quasi-bundling allocation; (b) quasi-flipping allocation; (c) crossing allocation. The interiors of some regions might be empty, but $R_\emptyset\neq \emptyset$ can be assumed w.l.o.g.
We marked the functions of critical values  $\psi(t_2,s)$ (Definition~\ref{def:critval}) for receiving task $1$  by broken lines.}
\label{fig:shapesPure}
\end{figure}

\paragraph{Additive players} 
It is known that in the case of two tasks, the regions in a WMON
allocation subdivide $R_{\geq 0}^2$ basically in three possible forms,
which are characteristic for the type of the whole
allocation-function $A$ (see Figure~\ref{fig:shapesPure}). The regions and their boundaries
determine the \emph{critical values} for $t_1$ (as function
of $t_2$ ) above which the $t$-player cannot get task 1, and
symmetrically for task 2. These critical value functions are
determined by the payment functions $P_\emptyset(s)=0,
P_1(s),P_2(s),P_{12}(s)$ for the fixed $s.$

\begin{definition} For given $s,$ we call the allocation $A[s]$
\begin{itemize}
\item \emph{quasi-bundling}, if there are at least two points $t\neq t'$ on the boundary of $R_{12}$ and $R_{\emptyset}$
\item \emph{quasi-flipping}, if there are at least two points $t\neq t'$ on the boundary of $R_{1}$ and $R_{2}$ 
\item \emph{crossing} otherwise (see  Figure~\ref{fig:shapesPure}). 
\end{itemize}
\end{definition}

\subsection{Mechanisms for submodular players}

We now introduce the types of WMON allocations that can occur when the
$t$-player is submodular and the $s$-player is additive and bounded
above by the value $B$. We denote these particular domains $\mathcal
T_t\times \mathcal T_s$ for two players and two tasks by
$V_\mathit{submod}\times V_{+,B},$ and refer the reader to the
appendix Section~\ref{sec:add} for detailed definitions, and also for
extension of our results to other domains. We assume that there always
exist high enough $t$ values so that the $t$-player receives no task,
i.e., that $R_\emptyset(s)\neq \emptyset$ for every $s.$\footnote{This
  assumption is without loss of generality for mechanisms with finite
  approximation of the makespan, even if some 2D cut mechanism of a
  WMON mechanism with more tasks and/or players is considered.}

\paragraph{1-dimensional mechanisms}

In a \emph{one-dimensional mechanism} at most two possible allocations are ever realized. Due to the assumption on $R_\emptyset,$ one of these must be the allocation $\emptyset$ to the $t$-player.
If the two occuring allocations (for the $t$-player) are $\emptyset$ and $12,$  we call the mechanism \emph{bundling mechanism}. The  other cases when the allocations to the $t$-player are $\emptyset$ and $1$ (or $\emptyset$ and $2$) are degenerate  \emph{task-independent} allocations, and can be defined similarly to bundling mechanisms. 

\begin{definition} 
 In a \emph{bundling mechanism} only  the allocations $\emptyset$ and $12$ can occur. There is an arbitrary, increasing\footnote{Throughout the paper, by 'increasing' we mean 'non-decreasing'; otherwise we say 'strictly increasing'.} function $\xi:[0,B)\rightarrow [0,\infty)$ so that if $t_{12}>\xi(s_1+s_2)$ then the $t$-player gets $\emptyset,$ and if $t_{12}<\xi(s_1+s_2)$ then the $t$-player gets $12.$ 

 If $\xi$ has a jump discontinuity in some point $s_1+s_2$ then the
 critical value may depend on the concrete $(s_1,s_2)$ with the given
 fixed sum, as long as it is between $\xi((s_1+s_2)^-)$ and
 $\xi((s_1+s_2)^+).$

\end{definition}

\paragraph{Relaxed affine minimizers}

\begin{definition} An allocation $A$ is an \emph{affine minimizer}, if there exist positive constants  per player $\mu_t$ and $\mu_s,$ and constants $\gamma_{\alpha} \in \mathbb R\cup\{-\infty,\infty\}$ per allocation (say, here $\alpha=\alpha_t$), so that for every input $(t,s)$ the allocation $A(t,s)$ minimizes over $$\mu_t t_{12}+\gamma_{12},\qquad \mu_t t_1+\mu_s s_2+\gamma_{1},\qquad \mu_s s_1+\mu_t t_2+\gamma_{2},\qquad \mu_s(s_1+s_2)+\gamma_{\emptyset}.$$
 \end{definition}

 \begin{definition}\label{dfn:relaxed-affine-minimizers}
   An allocation $A$ is a \emph{relaxed affine minimizer}, if there
   exist positive constants per player $\mu_t$ and $\mu_s,$ and
   constants $\gamma_{\alpha}$ per allocation $\alpha$ (of the
   $t$-player), furthermore an arbitrary increasing function
   $\xi:[0,\min(\gamma_{1},\gamma_{2})-\gamma_{\emptyset})\rightarrow
   [0,\infty)$ (if the interval
   $[0,\min(\gamma_{1},\gamma_{2})-\gamma_{\emptyset})$ is nonempty)
   with
   $\min(\gamma_{1},\gamma_{2})-\gamma_{12}=\xi(\min(\gamma_{1},\gamma_{2})-\gamma_{\emptyset})$,
   so that for every input $(t,s)$

\begin{itemize}
\item[$a)$] if $\mu_s\cdot(s_1+s_2)\geq \min(\gamma_{1},\gamma_{2})-\gamma_{\emptyset},$ the allocation $A(t,s)$ is that of an affine minimizer with the given constants 
\item[$b)$] if $\mu_s\cdot(s_1+s_2)\leq \min(\gamma_{1},\gamma_{2})-\gamma_{\emptyset},$ then if $\mu_t\cdot t_{12}>\xi(\mu_s(s_1+s_2))$ then the allocation for the $t$-player is $\emptyset$ and if $\mu_t\cdot t_{12}< \xi(\mu_s\cdot(s_1+s_2))$ then it is $12.$
\end{itemize}

\end{definition}

The boundary conditions in the above definition are such that the
affine minimizer part (a) fits with the bundling mechanism part (b)
and the resulting mechanism is truthful. See
Figure~\ref{fig:ex-relaxed-minimizer}, for an example of such a
mechanism in the additive domain.

\begin{figure}
\centering
\begin{tikzpicture}[scale=0.65]

\draw[->] (0,0) -- (8,0) node[anchor=north] {$t_1$};
\draw[->] (0,0) -- (0,8) node[anchor=east] {$t_2$};

\draw[very thick, blue] (0, 6) node[anchor=east,black] {$s_2$} -- (4, 6) -- (6, 4) -- (8, 4);
\draw[very thick, blue ] (6, 0) node[anchor=north,black] {$s_1$}-- (6, 4) -- (4, 6) -- (4, 8);

\draw[dashed, gray] (4,6) -- (4,0) node[anchor=north,black] {$s_1-1$};
\draw[dashed, gray] (6,4) -- (0,4) node[anchor=east,black] {$s_2-1$};

\draw[very thick, dashed, blue] (1.5,0) node[anchor=north,black] {$\sqrt{s_1+s_2}$} -- (0,1.5) node[anchor=east,black] {$\sqrt{s_1+s_2}$};

\end{tikzpicture}
\caption{An example of a relaxed affine minimizer, which shows the
  allocation of the $t$-player. The solid lines show the boundaries of
  the allocations for values of the $s$-player when $s_1+s_2\geq
  1$. The dashed lines show the allocation boundary when
  $s_1+s_2<1$. Sometimes we refer to this part of relaxed affine
  minimizers as ``bundling tail''.}
\label{fig:ex-relaxed-minimizer}
\end{figure}
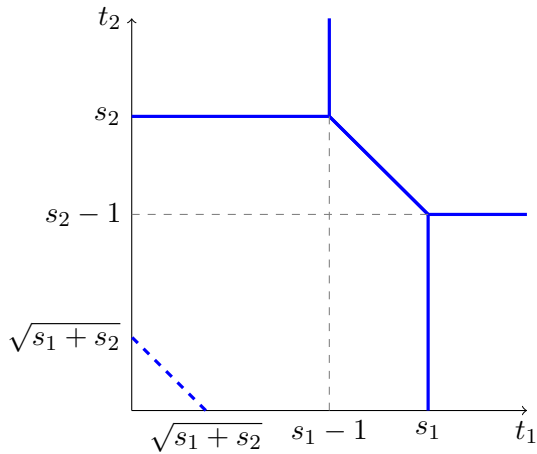

\paragraph{Constant mechanisms}

In a \emph{constant mechanism} the allocation is independent of the
bids of at least one of the players. This property can also be
interpreted as being an affine minimizer with multiplicative constant
$\mu=0.$ Due to the assumption on $R_\emptyset$ we only need to
consider constant mechanisms that are independent (at least) of the
$s$-player.

\subsection{The main characterization result}

The characterization that we use for the lower bound is captured by
the following theorem, whose proof can be found in the Appendix.

\begin{theorem}\label{theo:epschar} Every WMON allocation for two
  tasks and two players with bids $t\in V_{\mathit{submod}}$ and $s\in
  V_{+,B}$, where both tasks are always allocated, and
  $R_\emptyset(s)\neq \emptyset$ for every $s,$ is one of these three
  types: (1) relaxed affine minimizer, (2) one-dimensional mechanism,
  or (3) constant mechanism.  
\end{theorem}

In the Appendix we show that the same characterization (and hence our
lower bound) holds for WMON allocations for various domains of
valuations of the $t$-player and in particular for submodular
valuations which are arbitrarily close to additive, which we denote by
$V_{\epsilon}\cap V_\mathit{submod}$.

\section{Lower Bound}
\label{sec:lb}

\newcommand{\A}{\Theta}
\newcommand{\T}[1]{#1}

In this section, we give a proof of our main theorem
(Theorem~\ref{thm:main}). First in Section~\ref{sec:construction} we
describe the domain of instances that we use, and in
Section~\ref{sec:from-affine-minim} we use the characterization for
two machines and two tasks (Theorem~\ref{theo:epschar}) to establish
that the only interesting mechanisms are linear (see the subsection
for a precise definition), then in Section~\ref{sec:lin} we explore
the linearity property of mechanisms with bounded approximation ratio
to establish some useful locality lemmas that are eventually used in
Section~\ref{sec:proof-main} to complete the proof.

\subsection{The construction}
\label{sec:construction}

To prove the lower bound, we focus on the domain of $2(n-1)$ tasks and
$n$ players. Player $0$ is special and for convenience we use the
symbol $t$ for its values; sometimes we refer to it as the $t$-player. We
use the symbol $s$ for the values of the remaining players
$1,\ldots,n-1$, and sometimes we refer to them as the $s$-players.

The set of tasks
$M=\{\T1,\T{1'},\T2,\T{2'},\ldots, \T{n-1},\T{(n-1)'}\}$ is
partitioned in pairs and each pair $\{\T{i},\T{i'}\}$ is associated with
player $i$, $i=1,\ldots,n-1$.  We call the two tasks of each pair
$\{\T{i},\T{i'}\}$ \emph{twin
  tasks}. 

Let $v_i(S)$ denote the cost (valuation) of player $i$ when it takes
the subset $S\subseteq M$ of tasks.
\begin{itemize}
\item the cost of the $s$-players is \emph{additive}:
  $v_i(S)$ is additive for $i\geq 1$.
\item the cost of the $s$-players for tasks not in their associated
  pair is a sufficiently large fixed constant $\A\gg 4n^2$: for
  distinct $i,j\geq 1$, $v_i(\{\T{j}\})=v_i(\{\T{j'}\})=\A$.
\item the cost of the $t$-player for twin tasks is submodular: for
  every $i=1,\ldots,n-1$, the restriction of $v_0(S)$ to
  $S\subseteq\{\T{i},\T{i'}\})$ is submodular.
\item the cost of the $t$-player is \emph{additive across pairs}:
  $v_0(S)=\sum_{i=1}^{n-1} v_0(S\cap \{\T{i},\T{i'}\})$. Therefore
  $v_0(S)$ is \emph{submodular}, as the sum of submodular functions.
\end{itemize}

To simplify the notation we will denote an instance that satisfies the
above conditions by
\begin{align} \label{def:tt' instance}
(t_i,\,s_i,\,t_{i'},\,s_{i'},\,t_{i,i'})_{i=1}^{n-1},
\end{align}
where
\begin{itemize}
\item $s_i$ and $s_{i'}$ the cost of $i$-th $s$-player for the tasks
  in its associated pair $\{\T{i},\T{i'}\}$
\item $t_i$ and $t_{i'}$ the cost of the $t$-player when it takes
  \emph{only one} of the twin tasks $\{\T{i},\T{i'}\}$
\item $t_{i,i'}$ the cost of the $t$-player when it takes \emph{both
    twin tasks} $\{\T{i},\T{i'}\}$
\end{itemize}

With the exception of costs $t_{i,i'}$ an instance is captured by the
following matrix indexed by players and tasks, which shows the cost of
each when a player gets no other task.
\begin{align*}
  \left[
\begin{array}{c c c c c c c }
  t_1 & t_{1'} & t_2 & t_{2'} & \cdots & t_{n-1} & t_{(n-1)'}\\
  s_1 & s_{1'} & \A & \A & \cdots & \A & \A \\
  \A & \A & s_2 & s_{2'} & \cdots & \A & \A \\
  & &  & \vdots  &  &  &  \\
  \A & \A & \A & \A & \cdots & s_{n-1} & s_{(n-1)'} \\
\end{array}\right].
\end{align*}
If the valuations of all players were additive, this matrix would be
sufficient to determine the cost for all bundles. The instances we
consider have more general cost functions and include the submodular
valuation of twin tasks for the $t$-player. Since for two tasks, a
function is submodular if and only if it is subadditive, values
$t_{i,i'}$ satisfy
$$t_{i,i'}\leq t_i+t_{i'}.$$

It is useful to think of the value $\A$ as practically infinite, since
it is much larger that the other values. On the other hand, to prove
our main theorem in its generality, we need this value to be finite
and this complicates the characterization of truthful mechanisms.

We focus on instances that satisfy $s_i\in[0,1]$, $s_{i'}=n$,
$t_{i'}=0$, $t_{i,i'}=t_i+t_{i'}$ for all $i\in[n-1]$. Note the
subtlety here: while the domain contains instances with
$t_{i,i'}\leq t_i+t_{i'}$, for the proof of the lower bound, we
consider the subclass of additive instances. This should not be
surprising in the sense that lower bound proofs usually employ a
subclass of instances. Still, it raises the question whether we could
carry out the same proof in the additive domain. The answer is
negative, because the sets of mechanisms for additive and submodular
(subadditive) domains for two tasks are different; for example,
task-independent mechanisms are not truthful for submodular domains.

\begin{definition}[Restricted $(t,s)$ instance]
  Instances of the form
  $(t_i,\,s_i,\,t_{i'},\,s_{i'},\,t_{i,i'})_{i=1}^{n-1}$
  (Equation~\ref{def:tt' instance}) that satisfy
\begin{align}  \label{def:ts instance}
   s_i&\in[0,1], \quad s_{i'}=n, \quad t_{i'}=0, \quad
        t_{i,i'}=t_i+t_{i'}, & \text{for all $i\in[n-1]$,}
  \\
    & (t,s)=((t_i)_{i=1}^{n-1},(s_i)_{i=1}^{n-1}) \nonumber
\end{align}
will be called \emph{restricted $(t,s)$ instances}.
\end{definition}
Note that, unlike the other values, the $t_i$ values can be
arbitrarily high. This will be useful later
(Lemma~\ref{lem:linear}). Note also that the optimum makespan of every
restricted instance is at most 1. For these instances, any algorithm
with approximation ratio less than $n$, must allocate all the $\T{i'}$
tasks to the $t$-player and every $\T{i}$ task either to the
associated $s_i$-player or the
$t$-player. 

\subsection{From affine minimizers for twin tasks to linear mechanisms}
\label{sec:from-affine-minim}

In this section, we use the characterization of mechanisms
(Theorem~\ref{theo:epschar}) for \emph{each pair of twin tasks} in an
instance $(t_i,\,s_i,\,t_{i'},\,s_{i'},\,t_{i,i'})_{i=1}^{n-1}$, to
derive an essential property of mechanisms for the restricted $(t,s)$
instances. By the characterization, there are mechanisms of twin
tasks, such as the one-dimensional, constant, or bundling tails of
relaxed affine minimizers, that may have non-affine boundaries between
different allocations. In this section, we show that mechanisms with
non-affine boundaries are also excluded, when we additionally require
that the mechanisms have small approximation ratio---on whole
instances, not just on a pair of twin tasks.

First we define the notion of boundaries that we consider in proof of
the lower bound.
\begin{definition}\label{def:critval} For a given allocation
  algorithm, considered on restricted $(t,s)$ instances, we call
  $\psi_i(s_i,s_{-i},t_{-i})$ a \emph{critical value} or
  \emph{boundary} if the $t$-player receives task $i$ when
  $t_i<\psi_i(s_i,s_{-i},t_{-i}),$ and does not receive task $i$ when
  $t_i>\psi_i(s_i,s_{-i},t_{-i})$ (for example in
  Figure~\ref{fig:shapesPure}, $\psi_1$ is shown with dashed lines).
\end{definition}

We claim below that a truthful allocation with a reasonable
approximation ratio (say, of at most $n$) for the restricted $(t,s)$
instances satisfies the following important property:
\begin{definition}[Linear mechanisms]
  An allocation algorithm for the restricted $(t,s)$ instances
  (Equation~\ref{def:ts instance}), is called \emph{linear} if the
  critical values for task $i$, $i=1,\ldots,n-1$ are truncated affine
  functions in $s_i$. In particular when
  \begin{align*}
  \psi_i(s_i,s_{-i},t_{-i}) = \max(0,\, \lambda_i(s_{-i},t_{-i}) \, s_i \, +
  \, \gamma_i(s_{-i},t_{-i})),
  \end{align*}
  for some $\lambda_i(s_{-i},t_{-i}) > 0$. We call a mechanism
  linear, if it uses a linear allocation algorithm.
\end{definition}

\begin{lemma} \label{lem:linear} A mechanism for the instances
  $(t_i,\,s_i,\,t_{i'},\,s_{i'},\,t_{i,i'})_{i=1}^{n-1}$
  (Equation~\ref{def:tt' instance}) with approximation ratio less than
  $n$ must be linear for the restricted $(t,s)$ instances.
\end{lemma}
\begin{proof}
  By Lemma~\ref{lem:restriction}, the restriction of the mechanism for
  the twin tasks $\T{i}$, $\T{i'}$ is weakly monotone, and therefore
  by the characterization theorem (Theorem~\ref{theo:epschar}), it is
  either a relaxed affine minimizer, a one-dimensional, or a constant
  mechanism. We will show that if the approximation ratio is less than
  $n$, the mechanism cannot be one-dimensional, constant, or even
  relaxed affine minimizer with a bundling tail. The reason
  is that algorithms in all these cases are so inefficient that their
  makespan only for these two tasks is large enough to show a large
  approximation ratio for the entire instance.
  
  Let's first argue that it cannot be a one-dimensional or a constant
  mechanism. We fix all values of the other tasks as in restricted
  instances. By Lemma~\ref{lem:highapprox}, there are values for the
  twin tasks $\T{i}$, $\T{i'}$ for which the algorithm has makespan
  $\mu\geq\sqrt{\A}$ and approximation ratio at least $\sqrt{\A}$. Since
  the optimal makespan of the other tasks of a restricted instance is
  at most $1$, the approximation ratio is at least
  $\mu/(1+\mu/\sqrt{\A})\geq \sqrt{\A}/2\geq n$.

  As a result, the restriction to twin tasks $\T{i}$, $\T{i'}$ is a
  relaxed affine minimizer. We now argue that there should be no
  bundling part of the relaxed affine minimizer when the approximation
  ratio is less than $n$. Let us fix the values $s_j\in[0,1]$,
  $s_{j'}=n$ for all $j\in[n-1]$, and the values $t_j'=0$,
  $t_{j,j'}=t_j+t_{j'}$ for all $j\neq i$. We observe that the
  restriction of the mechanism in the cut $t_{-\{i,i'\}}$ (the
  2-dimensional partition when we also fix all other values except
  $t_i$ and $t_i'$) should definitely contain an area where the
  $t$-player gets task $i'$ but not $i$ (in particular, when
  $t_{i'}=\delta$, for some arbitrarily small $\delta$, and $t_i$ is
  very large). Otherwise the mechanism has approximation ratio at
  least $n$.
  
  This shows that in the domain of restricted $(t,s)$ instances, the
  mechanisms with approximation ratio less than $n$ are affine
  minimizers in every cut of pairs of twin tasks (even if in general,
  that is, for $s_{i'}<n$ they are relaxed affine minimizers).
  Therefore, the only mechanisms with approximation ratio less than
  $n$ are linear for the restricted $(s,t)$ instances.
\end{proof}

\subsection{Linearity}
\label{sec:lin}

From now on we focus on linear truthful mechanisms for restricted
$(t,s)$ instances. Establishing linearity of the boundaries
(Lemma~\ref{lem:linear}) is not directly useful, because the linear
coefficient of the boundary $i$ \emph{may depend on the other values
  of $t$}. The next lemma shows that this is not the case for the
scaling factor $\lambda_i$, although the $\gamma_i$ term may still
depend on the other values of $t$. Its proof is based on the
interaction of pairs of tasks $\T{i}$ and $\T{j}$. Note that these are
\emph{not} twin tasks but involve at least three players, the
$t$-player and the two associated $s$-players. Obtaining such
multiplayer statements is the bottleneck for a complete
characterization of mechanisms in multiplayer domains, and it is
perhaps the most crucial part of the proof.

\begin{lemma} \label{lem:linearity} For fixed $i$ and $j$, and for
  fixed $s_{-i}$, $t_{-ij}$, assume that $\psi_i(s_i,s_{-i},t_{-i})$
  is a truncated linear function of $s_i$, i.e.,
  $\psi_i(s_i,s_{-i},t_{-i})=\max(0,\, \lambda_i(s_{-i},t_{-i}) s_i+
  \gamma_i(s_{-i},t_{-i}))$, for some
  $\lambda_i(s_{-i},t_{-i})\geq 0$. Then $\lambda_i(s_{-i},t_{-i})$ is
  the same for all $t_j$ that satisfy $\psi_i(s_i,s_{-i},t_{-i})>0$.
\end{lemma}
\begin{proof}
  To keep the notation simple, we drop all values except for $s_i$ and
  $t_j$.  Let 
  $K(s_i)=\{t_j\,:\,\psi_i(s_i,t_j)>0\}$ be the interval of
  interest. Note first that when we increase $s_i$, the interval of
  interest can only expand, because $\psi_i(s_i,t_j)$ is
  non-decreasing in $s_i$.

  Within interval $K(s_i)$, function $\psi_i(s_i,t_j)$ is positive and
  therefore equal to $\lambda_i(t_j)s_i+\gamma_i(t_j)$. Furthermore,
  weak monotonicity implies that as a function of $t_j$,
  $\psi_i(s_i,t_j)$ is a piecewise linear function with derivative
  (slope) in $\{0,1,-1\}$ (see for example, Figure~\ref{fig:shapesPure}).

  The piecewise linear function $\psi_i(s_i,t_j)$ is
  differentiable\footnote{It is not hard to see that if
    $\psi_i(s_i,t_j)$ is differentiable then so are $\lambda_i(t_j)$
    and $\gamma_i(t_j)$; alternatively, one could use small
    differences instead of derivatives.}  everywhere except perhaps of
  the at most two break points. We have
\begin{align*}
  \frac{\partial \psi_i(s_i,t_j)}{\partial t_j}&= \frac{\partial
                                                 \lambda_i(t_j)}{\partial
                                                 t_j} s_i +  \frac{\partial
                                                 \gamma_i(t_j)}{\partial
                                                 t_j}. 
\end{align*}
If $\partial \lambda_i(t_j) /\partial t_j\neq 0$, then by varying
$s_i$, the slope cannot stay in $\{0,1,-1\}$. We conclude that
$\partial \lambda_i(t_j) /\partial t_j = 0$ in each piece, which shows
that $\lambda_i(t_j)$ is independent of $t_j$ within each piece and
therefore independent of $t_j$ everywhere. 
\end{proof}

The following corollary is an essential tool for establishing the
lower bound.

\begin{corollary} \label{cor:sliding-boundaries} The constant parts of
  the piecewise linear function $\psi_i(s_i,s_{-i},t_{-i})$, as a
  function of $t_j$, are independent of $s_j$, for $j\neq
  i$. 
\end{corollary}
\begin{proof}
  We fix all values except for tasks $i$ and $j$. By the previous
  lemma, when we change $s_j$, the boundary
  $\psi_j(s_j,s_{-j},t_{-j})$ is translated rectilinearly and
  therefore its break points remain the same. Since the constant parts
  of $\psi_i(s_i,s_{-i},t_{-i})$ are determined by the break points of
  $\psi_j(s_j,s_{-j},t_{-j})$, they also remain the same. See
  Figure~\ref{fig:upwards} for an illustration.
\end{proof}

\begin{corollary} \label{cor:sliding-boundaries2} The piecewise linear
  function $\psi_i(s_i,s_{-i},t_{-i})$ is either non-decreasing in
  $t_j$ or non-decreasing in $s_j$.
\end{corollary}
\begin{proof}
When we fix all other values and consider the $t_{-\{i,j\}}$ cut, it
is either quasi-flipping or crossing, in which case
$\psi_i(s_i,s_{-i},t_{-i})$ is non-decreasing in $t_j$, or
quasi-bundling, in which case 2D geometry shows that it is
non-decreasing in $s_j$ (when $s_j$ decreases the diagonal part shifts
downwards, and the rectilinear parts remain fixed). See
Figure~\ref{fig:upwards} for an illustration.

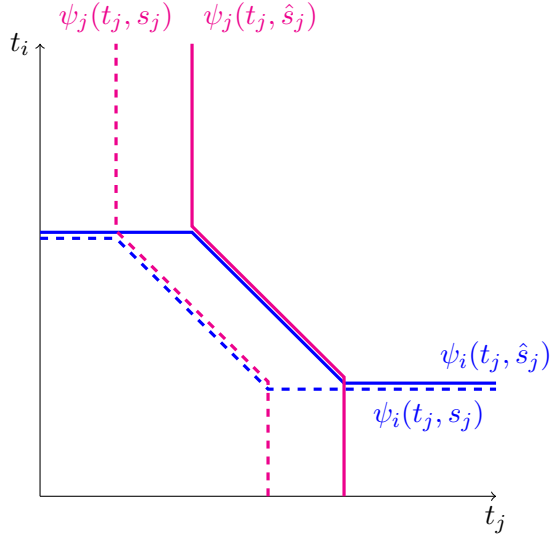
\begin{figure}
\centering
\begin{tikzpicture}

\draw[->] (0,0) -- (6,0) node[anchor=north] {$t_j$};
\draw[->] (0,0) -- (0,6) node[anchor=east] {$t_i$};
\draw[very thick, blue] (0,3.5) -- (2,3.5) -- (4,1.5) -- (6,1.5) node[anchor=south] {$\psi_i(t_j, \hat s_j)$} ;

\draw[very thick, magenta] (4,0) -- (4, 1.58) -- (2, 3.58) -- (2, 6)  node[anchor={south west}] {$\psi_j(t_j, \hat s_j)$} ;

\draw[very thick, magenta, dashed] (3,0) -- (3, 1.52) -- (1, 3.52) -- (1, 6) node[anchor=south] {$\psi_j(t_j, s_j)$} ;
 
\draw[very thick, blue, dashed] (0, 3.42) -- (1, 3.42) -- (3, 1.42) -- (6,1.42) node[anchor={north east}] {$\psi_i(t_j, s_j)$} ;

\end{tikzpicture}

\caption{Consider the dashed lines that show $\psi_i(t_j, s_j)$ and $\psi_j(t_j, s_j)$ (these functions may depend on other values but they play no role and we can ingore them). The regular lines show $\psi_i(t_j, \hat s_j)$ and $\psi_j(t_j, \hat s_j)$, for some $\hat s_j > s_j$. Notice that linearity implies that $\psi_i(t_j, \hat s_j)$ is a shift to the right of $\psi_i(t_j, s_j)$; in particular, the break point stay at the same height. As a consequence, the blue horizontal parts remain at the same height and the blue diagonal part shifts to the right. Therefore $\psi_i(t_j, \hat s_j) \geq  \psi_i(t_j, s_j)$.}
\label{fig:upwards}
\end{figure}

\end{proof}

\subsection{The main theorem}
\label{sec:proof-main}

\def\kkk{n-1}

We now have most of the ingredients to prove next theorem which
directly implies the main lower bound Theorem~\ref{thm:main}).

\begin{theorem} \label{thm:star} The approximation ratio of linear
  truthful algorithms on restricted $(s,t)$ instances
  (Equation~\ref{def:ts instance}) with $n$ machines is at least
  $\sqrt{n-1}$.
\end{theorem}

To prove this theorem, we fix some linear truthful algorithm and we
focus on a particular set of instances, which we call
$s$-inefficient. The set of $s$-inefficient instances depends on the
linear algorithm and each instance consists of two types of tasks:
either a task is unimportant (i.e., it has very small value for one of
the machines), or its $s$ value is significantly higher than its $t$
value, yet the algorithm allocates the task to the $s$-player.

\begin{definition}[$s$-inefficient instances] Let's call a task $i$
  {\em trivial} when either $s_i=0$ or $t_i\in (0,\delta_0]$, for some
  fixed (sufficiently small) $\delta_0$. We call a restricted $(t,s)$
  instance \emph{$s$-inefficient} for a mechanism, if it contains at
  least one non-trivial task, every non-trivial task $i$ satisfies
  $s_i/t_i > \sqrt{\kkk}$, and the mechanism allocates all non-trivial
  tasks to the $s$-player.
\end{definition}

\emph{The heart of the proof is to show that if the set of
  $s$-inefficient instances is non-empty, there exists an
  $s$-inefficient instance with exactly one non-trivial task}. From
this, it immediately follows that the algorithm has approximation
ratio at least $\sqrt{\kkk}$.

However, it may be that for a given linear truthful algorithm there
are no $s$-inefficient instances. But then we can use weak
monotonicity to easily derive a $\sqrt{\kkk}$ lower bound on the
approximation ratio as the following lemma shows.
\begin{lemma}
  If for a given truthful algorithm the set of $s$-inefficient instances
  is empty, then its approximation ratio is at least $\sqrt{\kkk}$.
\end{lemma}
\begin{proof}
  Towards a contradiction, consider an algorithm that has
  approximation ratio $\sqrt{\kkk}-\delta$, for some $\delta>0$, for
  which the set of $s$-inefficient instances is empty. We consider the
  instance with $t_i=\alpha=1/\sqrt{\kkk}-\delta/n$ and $s_i=1$,
  for all $i\in[\kkk]$.
\begin{align}\label{eq:101}\left[\begin{array}{cccc}
\alpha & \alpha & \cdots & \alpha \\
1 & & &\\
& 1 & &\\
&&\vdots&\\
&&& 1
\end{array}\right]\end{align}
Note that at least one task is allocated to the $s$-players, because
if all tasks are assigned to the $t$-player, the makespan is
$(\kkk)\alpha=\sqrt{\kkk}-(\kkk)\delta/n>\sqrt{\kkk}-\delta$, the optimum
makespan is $1$, and the approximation ratio is strictly greater than
$\sqrt{\kkk}-\delta$.

Now for every task $i$ which is allocated to the $t$-player, we lower
its value from $\alpha$ to some small value in $(0,\delta_0]$ and
increase the $t$ value of every other task to
$\alpha+\delta/(2n)<1/\sqrt{\kkk}$. By weak monotonicity
(Lemma~\ref{lemma:tool}), the allocation of the tasks for the
$t$-player remains the same. Since at least one task is allocated to
the $s$-players, we end up with an $s$-inefficient instance, a
contradiction.
\end{proof}

The proof of the next lemma (Lemma~\ref{lem:critical}), which provides
a very useful property of linear mechanisms, is the most critical part
of the proof. It shows that linear weakly monotone algorithms satisfy
a locality property, that bears some resemblance to the locality
property in~\cite{NR01}. But unlike~\cite{NR01}, our proof does not
assume this property, but it derives it from weak monotonicity for the
special class of instances that we consider.

\begin{lemma} \label{lem:critical}
  If for a given linear truthful algorithm the set of $s$-inefficient
  instances is non-empty, then there is an $s$-inefficient instance
  with exactly one non-trivial task.
\end{lemma}
\begin{proof}
  Fix a linear truthful algorithm and consider an $s$-inefficient instance
  $(t,s)$ with the minimum number of non-trivial tasks. If the number
  of non-trivial tasks is at least two, let us assume without loss of
  generality that tasks $1$ and $2$ are non-trivial. We will derive a
  contradiction by reducing the number of non-trivial tasks.

  Consider the boundary function of the first task,
  $\psi_1(s_1,s_{-1},t_{-1})=\max(0, \lambda_1(s_{-1},t_{-1})
  s_1+\gamma_1(s_{-1},t_{-1})$, for some positive $\lambda_1$. The
  crux of the matter is that we can reduce the value of $s_1$ to $0$
  or $t_1$ to at most $\delta_0$, and guarantee that the $t$-player will
  keep not getting the second task. This guarantees that the second
  task is non-trivial and is given to the $s$-player, while the first
  task becomes trivial.

  If $\psi_2(s_2,s_{-2},t_{-2})$ as a function of $t_1$ is
  non-decreasing (i.e., the $t_{-\{1,2\}}$ cut is quasi-flipping or
  crossing), we set $t_1^{*}\in(0,\delta_0]$. Since
  $\psi_2(s_2,s_{-2},t_{-2})$ is non-decreasing in $t_1$
  (Corollary~\ref{cor:sliding-boundaries2}), in the new instance the
  second task is still allocated to the $s$-player
  (Figure~\ref{fig:shift-quasiflipping}).

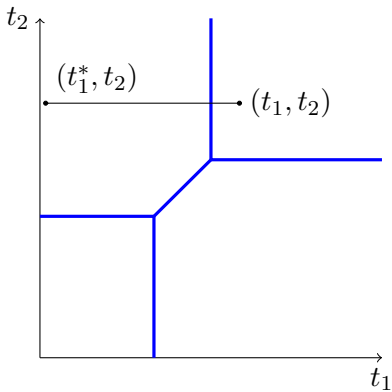
\begin{figure}
\centering
\begin{tikzpicture}[scale=0.75]

\draw[->] (0,0) -- (6,0) node[anchor=north] {$t_1$};
\draw[->] (0,0) -- (0,6) node[anchor=east] {$t_2$};
\draw[very thick, blue] (0,2.5) -- (2,2.5) -- (2,0) ;
\draw[very thick, blue] (2,2.5) -- (3,3.5) ;
\draw[very thick, blue] (3,6) -- (3,3.5) -- (6,3.5);

\filldraw 
  (3.5,4.5) circle (1pt) node [anchor=west] {$(t_1,t_2)$} --
  (0.1,4.5) circle (1pt) node [anchor={south west}] {$(t_1^*,t_2)$};
\end{tikzpicture}
\caption{When the cut is quasiflipping and neither task is given to
  the $t$-player in the allocation for $(t_1,t_2)$, then the second
  task is still not allocated to the $t$-player in the allocation for
  $(t_1*\approx 0,t_2)$.}
\label{fig:shift-quasiflipping}
\end{figure}

Otherwise (i.e., the $t_{-\{1,2\}}$ cut is quasi-bundling), by
Corollary~\ref{cor:sliding-boundaries}, function
$\psi_2(s_2,s_{-2},t_{-2})$ is non-decreasing in $s_1$. In this case,
we change the instance as follows (see
Figures~\ref{fig:shift-quasibundling} for an illustration of the first
two cases):
  \begin{enumerate}
  \item[a)] if $\gamma_1(s_{-1},t_{-1})\geq 0$, we set $s_1^{*}=0$ and $t_1^{*}=t_1$, 
  \item[b)] otherwise, if $\psi_1(s_1,s_{-1},t_{-1})>0$, we set
    $s_1^{*}=-\gamma_1(s_{-1},t_{-1})/\lambda_1(s_{-1},t_{-1})$ and
    $t_1^{*}\in (0,\delta_0]$, small enough to make task $1$ trivial,
  \item[c)] and if $\psi_1(s_1,s_{-1},t_{-1})=0$, we set $s_1^{*}=s_1$ 
    and $t_1^{*}\in (0,\delta_0]$, small enough to make task $1$ trivial.
  \end{enumerate}
  In the first two cases, we lower the $s_1$ value to $s_1^*$ until
  either $s_1^*$ becomes $0$ (case a), or
  $\psi_1(s_1^*,s_{-1},t_{-1})$ becomes $0$ and then we set
  $t_1=t_1*\approx 0$ (case b). The third case is when
  $\psi_1(s_1,s_{-1},t_{-1})$ is already $0$.

  In all cases, the first task becomes a trivial task with the new
  values. Furthermore, in all cases the new value of $s_1$ is not
  greater than the original value: $s_1^{*}\leq s_1$. This is clearly
  true in the first and third case. To see that this is true in the
  second case, observe that the boundary $\psi_1(s_1,s_{-1},t_{-1})$,
  which is a non-decreasing function on $s_1$, moved from a positive
  value to $\psi_1(s_1^{*},s_{-1},t_{-1})=0$.
  
  We now argue that \emph{the second task is still given to the
    $s$-player after the change}. The argument is based on
  Corollary~\ref{cor:sliding-boundaries}, which guarantees that the
  changes can only shift the boundary $\psi_2(s_2,s_{-2},t_{-2})$
  rectilinearly: in Figure~\ref{fig:shift-quasibundling}, the slanted
  part moves only horizontally.

  For case (a), by Corollary~\ref{cor:sliding-boundaries}, the
  boundary $\psi_2(s_2,s_{-2},t_{-2})$ did not increase and therefore
  the second task is still given to the $s$-player (left part of
  Figure~\ref{fig:shift-quasibundling}).  For the other two cases,
  this is not sufficient because $t_1$ changed and therefore $t_{-2}$
  changed as well. For case (b), the change shifts the slanted
  boundary (right part of Figure~\ref{fig:shift-quasibundling}). In
  its new position, it crosses the boundary of the positive orthant at
  $(0,t_2)$ which is dominated by the point $(t_1^*,t_2)$. Therefore,
  the $t$-player gets neither task, and shows that the $s$-player gets
  the second task. Case (c) is simpler and similar to the second one;
  the difference is that the boundary starts at the leftmost position
  and it does not move at all (in the right part of
  Figure~\ref{fig:shift-quasibundling}, the solid and dashed lines
  coincide and they may closer to the beginning of axes).

\begin{figure}
\centering
\begin{tikzpicture}

\draw[->] (0,0) -- (6,0) node[anchor=north] {$t_1$};
\draw[->] (0,0) -- (0,6) node[anchor=east] {$t_2$};
\draw[very thick, blue] (0,3.5) -- (2,3.5) -- (4,1.5) -- (6,1.5) node[anchor=south] {$\psi_2(t_1, s_1)$} ;
\draw[very thick, blue ] (4,0) -- (4, 1.5) -- (2, 3.5) -- (2, 6) node[anchor=west] {$\psi_1(t_2, s_2)$} ;

\filldraw 
  (3.5,3) circle (1pt) node [anchor=west] {$(t_1,t_2)$} 
  ;

\draw[very thick, blue, dashed] (0, 3.45) -- (1, 3.45) -- (3, 1.45) -- (6,1.45) node[anchor=north] {$\psi_2(t_1, s_1^*=0)$} -- (6, 1.45) ;
\draw[very thick, blue, dashed]  (3,0) -- (3, 1.45) ;
\draw[very thick, blue, dashed]  (1, 3.45) -- (1, 6) ;

\end{tikzpicture}
\begin{tikzpicture}

\draw[->] (0,0) -- (6,0) node[anchor=north] {$t_1$};
\draw[->] (0,0) -- (0,6) node[anchor=east] {$t_2$};
\draw[very thick, blue] (0,3.5) -- (2,3.5) -- (4,1.5) -- (6,1.5) node[anchor=south] {$\psi_2(t_1, s_1)$} ;
\draw[very thick, blue ] (4,0) -- (4, 1.5) -- (2, 3.5) -- (2, 6) node[anchor=west] {$\psi_1(t_2, s_2)$} ;

\filldraw 
  (3.5,3) circle (1pt) node [anchor=west] {$(t_1,t_2)$} --
  (0.1,3) circle (1pt) node [above right=-0.1] {$(t_1^*\approx 0,t_2)$}
  ;

\draw[very thick, blue, dashed] (0,3) -- (1.55, 1.45) -- (2.75, 1.45) node[anchor=north] {$\psi_2(t_1, s_1^*)$} -- (6, 1.45) ;
\draw[very thick, blue, dashed] (1.55, 1.45) -- (1.55, 0) ;
  
\end{tikzpicture}

\caption{Left side case (a), when $s_i*$ decreases to $0$: the
  boundary $\psi_2(t_1,s_1^*)$ is lower than the boundary
  $\psi_2(t_1,s_1^*)$. Both tasks are still allocated to the
  $s$-player. Right side case (b), when the $s_i$ decreases to
  $s_1^{*}=-\gamma_1(s_{-1},t_{-1})/\lambda_1(s_{-1},t_{-1})$ and
  $t_1$ changes to $t_1^*\approx 0$. The new $(t_1^*,t_2)$ remains in
  the region in which both tasks are allocated to the $s$-player.}
\label{fig:shift-quasibundling}
\end{figure}
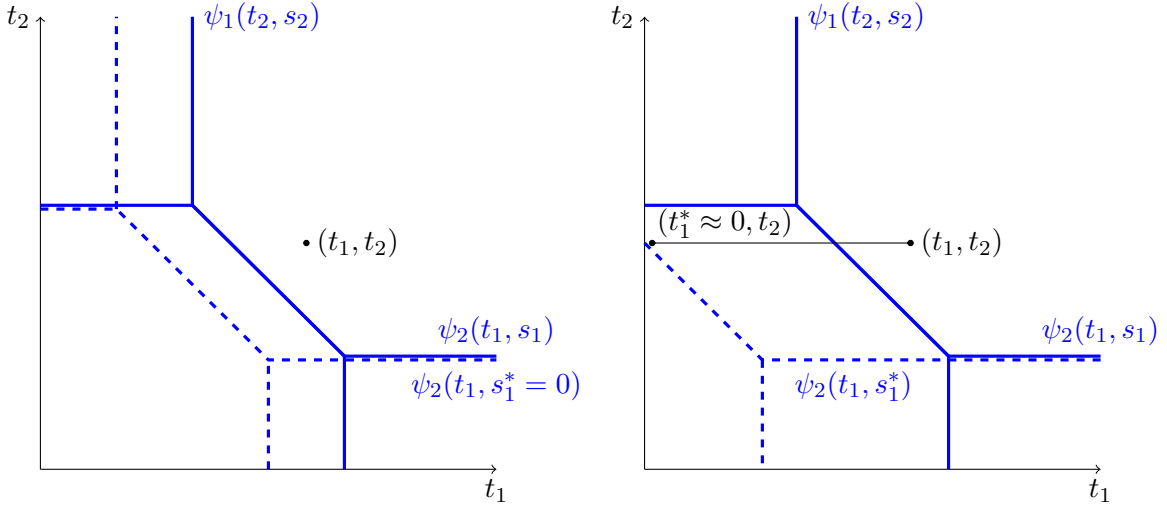

  The change of values of the first task did not change the allocation
  of the second task, but it may have changed the allocation of the
  remaining non-trivial tasks. But we can use weak monotonicity to
  further change the instance to obtain an $s$-inefficient instance. If
  some non-trivial tasks changed allocation and were given to the $t$
  player, we reduce their values to $0$. This will make them trivial
  tasks and, by weak monotonicity (Lemma~\ref{lemma:tool}), preserve
  the allocation of the first player for the other non-trivial
  tasks\footnote{Strictly speaking, to guarantee that the allocation
    remains the same for the non-trivial tasks, we need to increase
    slightly their $t$ value. This is possible without violating the
    constraint $s_i/t_i>\sqrt{n-1}$, which shows that the resulting
    instance is also $s$-inefficient.}.
  The resulting instance is $s$-inefficient and has fewer non-trivial
  tasks, a contradiction.
\end{proof}

The proof of the main result of this section follows directly from the
last two lemmas.

\begin{proof}[Proof of Theorem~\ref{thm:star}]
  By the last two lemmas, either a linear truthful algorithm has
  approximation ratio at least $\sqrt{n-1}$, or there exists an
  $s$-inefficient instance with one non-trivial task. The
  approximation ratio of such an instance is trivially at least
  $\sqrt{n-1}$, and the proof is complete.
\end{proof}

\bibliographystyle{plain} 

\newpage

\appendix

\part*{Appendix: All truthful mechanisms for 2 tasks and 2 machines}

\section{Additive valuations for 2 tasks and 2 machines}
\label{sec:add}

\subsection{Basic concepts and notation}

In this section we characterize WMON allocations for two tasks and two players (called $t$-player, and $s$-player) with additive valuations $t=(t_1, t_2)$, and $s=(s_1,s_2)$, so that both $s_1$ and $s_2$ are bounded by an arbitrarily large but fixed value $B.$

Let $(A,P)$ be  a truthful mechanism, where $A$ is the WMON allocation function, and $P$ denotes the payment function. For input $(t,s)$ the allocation is $A(t,s).$
Since we have only tasks 1 and 2, in $A(t,s)$ we can denote the allocation to one of the players as $\alpha_t,\alpha_s\in\{12,1,2,\emptyset\}.$ We often abuse notation and identify $A(t,s)$ by the allocation $\alpha_t$ of the $t$-player; e.g., we say that the allocation $A(t,s)$ is $12,$ meaning that $\alpha_t=12,$ and $\alpha_s=\emptyset.$

For given $s\in[0,B)\times [0,B)$ the allocation for the $t$-player as
\emph{function of his bids $(t_1,t_2)$} is denoted by $A[s],$ and
symmetrically $A[t]$ is an allocation function for the $s$-player.
For $\alpha_t\in\{12,1,2,\emptyset\},$ the allocation \emph{regions}
$R_{\alpha_t}(s)\subseteq \mathbb R_{\geq 0}^2$ of $A[s]$ are defined
to be the \emph{interior} (wrt. $\mathbb R_{\geq 0}^2$) of the set of
all $t$ values such that $A(t,s)=\alpha_t.$ (For the $s$-player we
denote the respective regions by $R_{\alpha_s}^s(t).$ ) We assume that
for every $s\in [0,B)\times [0,B)$ there exist $t$ values so that the
$t$-player receives no task, i.e., that $R_\emptyset(s)\neq \emptyset$
for every $s.$ Note that we do \emph{not} assume the same for the
$s$-player, because his bids are bounded. We remark that the
assumption $R_\emptyset(s)\neq \emptyset$ does not restrict the
\emph{types} of possible WMON mechanisms, but it simplifies the
characterization to a large extent.

It is known that in the case of  two tasks, the regions in a WMON allocation subdivide $R_{\geq 0}^2$
basically in three  possible forms, which will turn out to be  characteristic for the type of the whole allocation-function $A.$ (see Figure~\ref{fig:shapes}) The regions and their boundaries determine the \emph{critical values} for $t_1$ (in fact, as function of $t_2$ ) above which the $t$-player cannot get task 1, and symmetrically for task 2. These critical value functions are ultimately determined by the payment functions $P_\emptyset(s)=0, P_1(s),P_2(s),P_{12}(s)$ for the fixed $s.$

For any given $s,$ we assume w.l.o.g. \emph{virtual payments} for the $t$-player, which are the unique normalized and monotone (as set functions) payments for a given allocation $A[s]$ (see also Section~\ref{sec:nonadd}):

\begin{definition}\label{def:virtadd} Let $(A,\hat P)$  be a truthful mechanism, and  $s=(s_1, s_2)$ be fixed. The payments occuring in this definition are all defined for this given $s,$ so we omit $(s)$ from the notation.  Since $R_\emptyset\neq \emptyset,$ we can assume w.l.o.g. that the payment function $\hat P$ is normalized, that is, $\hat P_{\emptyset}=0.$  We define the \emph{virtual payments} as:
$$P_\emptyset=0,$$
$$P_1=\max\{\hat P_1,0\},\quad P_2=\max\{\hat P_2,0\},$$
$$P_{12}=\max\{\hat P_{12},P_1,P_2\}.$$

\end{definition}

The virtual payments are always well-defined, non-negative, and truthful payments to the $t$-player  for the given allocation $A[s].$

For fixed $s,$ the virtual payments determine uniquely the allocation and vice-versa. The value $P_1$ determines the position of the vertical boundary/critical value between $R_\emptyset$ and $R_1.$ Since $R_\emptyset$ is nonempty, this boundary always exists; if $R_1=\emptyset,$ then $P_1=0.$ Analogously, $P_2$ is the position of the horizontal boundary  between $R_\emptyset$ and $R_2;$  furthermore $P_{12}-P_2$ (resp. $P_{12}-P_1$) is the position of the vertical (horizontal) boundary between $R_{12}$ and $R_2$ (resp. $R_{12}$ and $R_1$) if these regions are nonempty. In the proof below, for given $s$ we also use the short notation
$$f'=P_1,\quad g'=P_2,$$ and if the respective boundary exists then also $$f=P_{12}-P_2,\quad g=P_{12}-P_1.$$

\begin{definition} For given $s,$ we call the allocation $A[s]$
\begin{itemize}
\item \emph{quasi-bundling}, if there are at least two points $t\neq t'$ on the boundary of $R_{12}$ and $R_{\emptyset}$($\Leftrightarrow P_1<P_{12}-P_2$)
\item \emph{quasi-flipping}, if there are at least two points $t\neq t'$ on the boundary of $R_{1}$ and $R_{2}$ ($\Leftrightarrow P_1>P_{12}-P_2$).

\item \emph{crossing} otherwise ( $\Leftrightarrow P_1=P_{12}-P_2$) (see also Figure~\ref{fig:shapes})
\end{itemize}
\end{definition}

\begin{figure}[t]
\centerline{\includegraphics[height=4.5cm]{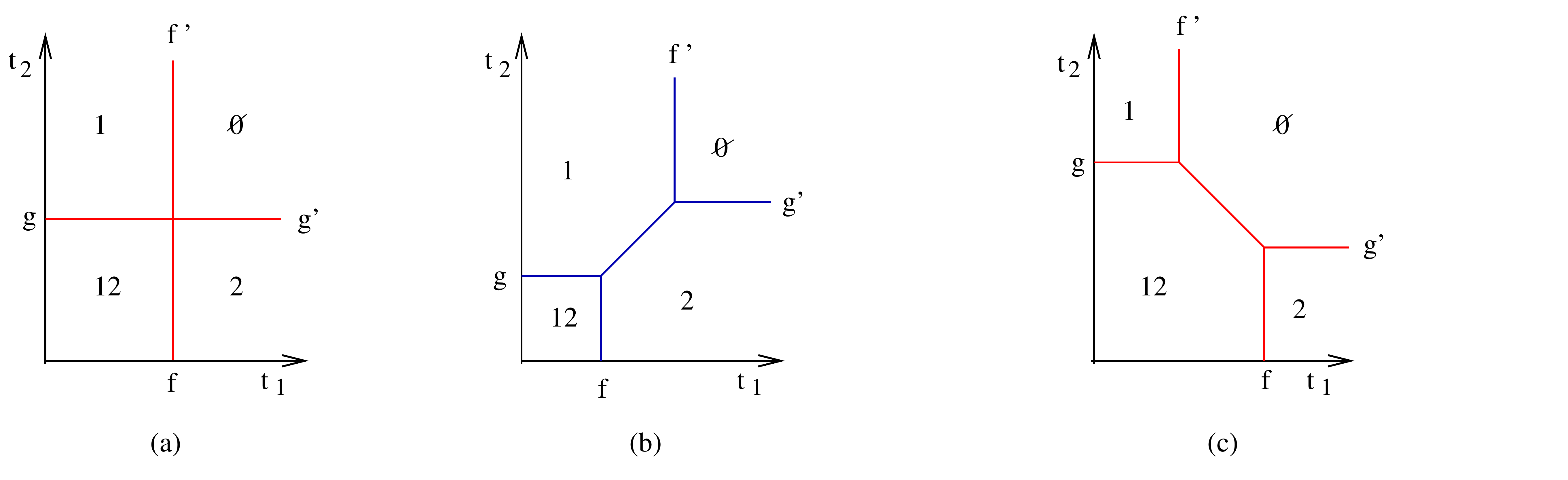}}
\caption{The allocations to a single player depending on his own 2-dimensional bid vector partition the bid-space according to one of these three shapes (where the interior of some region(s) different from $R_\emptyset$ may be empty). The allocation $A[s]$ is called (a) \emph{crossing}, (b) \emph{quasi-flipping}, or (c) \emph{quasi-bundling}, respectively. }
\label{fig:shapes}
\end{figure}

Next we introduce the four different types of WMON allocations that can occur in the above setting. After that we prove the characterization.\footnote{Observe that there exist mechanisms that adhere to more than one of the given types.}

\subsubsection{Relaxed task-independent allocations}

An allocation function $A$ is \emph{task-independent} if, for $i\in \{1,2\},$ the allocation of task $i$ depends only on the input values $s_i$ and $t_i.$  For the $t$-player, the critical-value of $t_1,$ i.e., the lowest value, above which $t_1$ does not get task $1,$ is determined by an \emph{arbitrary}, increasing function $\phi:[0,B)\rightarrow [0,\infty)$ of $s_1.$
Analogously we can define the critical-value function $\eta(s_2)$ as the lowest value, depending on $s_2,$ above which $t_2$ does not get task $2.$\footnote{The tie-breaking (i.e., when e.g., $t_1=\phi(s_1)$) may in all cases depend on the other variables ($s_2$ and $t_2$)
.} Geometrically, in a task-independent mechanism, the allocations $A[s]$ and $A[t]$ of both players are always crossing. 
In a \emph{relaxed task-independent} mechanism the latter property is fulfilled in all but countably many $s$ (resp. $t$) points, in which both $\phi$ and $\eta$ (resp. $\phi^{-1}$ and $\eta^{-1}$) have a jump discontinuity:

\begin{definition} An allocation  is a \emph{relaxed task-independent allocation} if there exist arbitrary, increasing critical-value functions for the $t$-player: $\phi(s_1)$ for task $1,$ and $\eta(s_2)$ for task $2,$ so that for every $s\in [0,B)\times [0,B)$

\begin{itemize}
\item[--] if $\phi()$ is continuous in the point $s_1$ OR $\eta()$ is continuous in the point $s_2$ then the allocation of the two tasks is independent and adheres to the critical values
$\phi(s_1)$ and $\eta(s_2),$ respectively; 
(note however that if, e.g.,  $\eta$ has a jump discontinuity in $s_2,$ then the critical value of for task $2$ can depend on $s_1,$ as long as it is between $\eta(s_2^-)$ and $\eta(s_2^+)$); 

\item[--] if $\phi()$ has a jump discontinuity in the point $s_1$ AND $\eta()$ has a jump discontinuity in $s_2,$ then the allocation $A[s]$  of the $t$-player can be an arbitrary  crossing, quasi-flipping, or quasi-bundling allocation, such that all critical values for task $1$ are at least $\phi(s_1^-)$ and at most $\phi(s_1^+),$ and similarly for task $2.$
\end{itemize}

Symmetric statements hold for the $s$-player with the critical-value-functions $\phi^{-1}(t_1)$ and $\eta^{-1}(t_2).$
In summary, every relaxed task-independent allocation is identical with a task-independent allocation on $t_i\in [0,\infty)\setminus T_i, \,s_i\in [0,B)\setminus S_i, $ where the $T_1, T_2, S_1, S_2$ are  countable sets.

\end{definition}

\subsubsection{1-dimensional mechanisms}

In a \emph{one-dimensional mechanism} at most two possible allocations are ever realized. The critical-value function between the two possible allocations  is an arbitrary increasing function of the respective input variables. 

Because of the assumption that $R_\emptyset(s)\neq \emptyset$ for every $s,$ one of the occuring allocations of the $t$-player must be the allocation $\emptyset.$
If the two occuring allocations (for the $t$-player) are $\emptyset$ and $12,$  we call the mechanism \emph{bundling mechanism}. In the other cases when $1,$ and $\emptyset;$ or when $2,$ and $\emptyset$ are the possible allocations of the $t$-player, the mechanism is identical with a degenerate task-independent mechanism, where $\phi\equiv 0$ or $\eta\equiv 0.$ Therefore we define only the case of a bundling mechanism in detail. We note that any bundling mechanism can be considered as a relaxed affine minimizer, with $\gamma_1=\infty$ and $\gamma_2=\infty,$ but it will be convenient to treat bundling mechanisms separately. 

\begin{definition} 
 In a \emph{bundling mechanism} only  the allocations $\emptyset$ and $12$ can occur. There is an arbitrary, increasing function $\xi:[0,B)\rightarrow [0,\infty)$ so that $\xi(s_1+s_2)$ is the lowest value such that if $t_1+t_2>\xi(s_1+s_2)$ then the $t$-player gets $\emptyset.$ 

If $\xi$ has a jump discontinuity in some point $s_1+s_2$ then the critical value may depend on the concrete $(s_1,s_2)$ with the given fixed sum, as long as it is between $\xi((s_1+s_2)^-)$ and $\xi((s_1+s_2)^+).$

\end{definition}

\subsubsection{Relaxed affine minimizers}

\begin{definition} An allocation $A$ is an \emph{affine minimizer}, if there exist positive constants  per player $\mu_t$ and $\mu_s,$ and constants $\gamma_{\emptyset},\gamma_{12},\gamma_{1},\gamma_{2}\in \mathbb R\cup\{-\infty,\infty\}$ per allocation (say, of the $t$-player), so that for every input $(t,s)$ the allocation $A(t,s)$ minimizes over $$\mu_t(t_1+t_2)+\gamma_{12},\qquad \mu_t\cdot t_1+\mu_s\cdot s_2+\gamma_{1},\qquad \mu_s\cdot s_1+\mu_t\cdot t_2+\gamma_{2},\qquad \mu_s(s_1+s_2)+\gamma_{\emptyset}.$$
 \end{definition}

In some affine minimizers, if $s_1+s_2$ and $t_1+t_2$ are both 'small', then the mechanism locally 'looks like' a bundling mechanism, with only the regions $R_{\emptyset}$ and $R_{12}$ for both players. If this is the case for \emph{all} $(s_1,s_2)$ and \emph{all} $(t_1,t_2)$ with  sums $s_1+s_2<D$ and $t_1+t_2<C$ for some given $C,D>0,$ then the mechanism becomes locally a bundling mechanism. This is formulated more precisely in the next definition:

\begin{definition} An allocation $A$ is a \emph{relaxed affine
    minimizer}, if there exist positive constants  per player $\mu_t$
  and $\mu_s,$ and  constants $\gamma_{\alpha}$ per allocation
  $\alpha$ (of the $t$-player),
furthermore an arbitrary increasing function
   $\xi:[0,\min(\gamma_{1},\gamma_{2})-\gamma_{\emptyset})\rightarrow
   [0,\infty)$ (if the interval
   $[0,\min(\gamma_{1},\gamma_{2})-\gamma_{\emptyset})$ is nonempty)
   with
   $\min(\gamma_{1},\gamma_{2})-\gamma_{12}=\xi(\min(\gamma_{1},\gamma_{2})-\gamma_{\emptyset})$,
   so that for every input $(t,s)$

\begin{itemize}
\item[$a)$] if $\mu_s\cdot(s_1+s_2)\geq \min(\gamma_{1},\gamma_{2})-\gamma_{\emptyset},$ the allocation $A(t,s)$ is that of an affine minimizer with the given constants 
\item[$b)$] if $\mu_s\cdot(s_1+s_2)\leq \min(\gamma_{1},\gamma_{2})-\gamma_{\emptyset},$ then  if $\mu_t\cdot(t_1+t_2)>\xi(\mu_s\cdot(s_1+s_2))$ then the allocation for the $t$-player is $\emptyset$ and if $\mu_t\cdot(t_1+t_2)\leq \xi(\mu_s\cdot(s_1+s_2))$ then it is $12.$
\end{itemize}

\end{definition}

The condition
$\min(\gamma_{1},\gamma_{2})-\gamma_{12}=\xi(\min(\gamma_{1},\gamma_{2})-\gamma_{\emptyset})$
is required so that the affine minimizer (a) is appropriately `glued'
to the bundling mechanism (b).

\subsubsection{Constant mechanisms}

In a \emph{constant mechanism} the allocation is independent of the bids of at least one of the players. This property can also be interpreted as being an affine minimizer with multiplicative constant $\mu=0.$ Since we are interested in non-constant critical value functions, it will be convenient to treat constant mechanisms separately from affine minimizers.

\subsection{The main result}

\begin{theorem}\label{theo:addchar} Every WMON allocation for two tasks and two additive players with bids $t\in (0,\infty)\times (0,\infty)$ and $s\in(0,B)\times (0,B)$, where both tasks are always allocated,  and $R_\emptyset(s)\neq \emptyset$ for every $s,$ is one of these four types:  (1) relaxed affine minimizer, (2) relaxed task-independent mechanism (3) one-dimensional mechanism, or (4) constant mechanism.
\end{theorem}

\subsection{Proof of Theorem~\ref{theo:addchar}}

\begin{notation}
For the critical values (boundary positions) of the $t$-player for given $s\in [0,B)\times[0,B)$ we use the notation $f'(s)=P_1(s),$ $g'(s)=P_2(s)$ ($s$ is omitted from the notation if it is clear from the context).
\end{notation}

The critical values $f',g'\in [0,\infty)$  are always defined and finite, since $f'=\infty$ or $g'=\infty$ would mean $R_\emptyset=\emptyset.$

We prove next, that $f'$ is an increasing function of $s_1,$ furthermore it is continuous in almost every (i.e., up to at most countably many) point $s_1\in[0,B),$ and in every  $s_1$ where $f'$ is continuous, it is independent of $s_2.$  Symmetric statements hold for $g'.$

\begin{lemma}\label{lem:cont} $f'$ is monotone increasing in $s_1$ in the following strong sense: Let $s_2,s_2'$ be arbitrary. If $s_1< s_1'$  then $f'(s_1,s_2)\leq f'(s_1',s_2').$ \end{lemma}
\begin{proof} Let $s=(s_1,s_2),$ and $s'=(s_1',s_2'),$ and assume for contradiction that $f'(s)> f'(s').$ Then there exists a $\hat t$ (with high enough $\hat t_2$) so that $\hat t\in R_{1}(s)$ but $\hat t\in R_\emptyset(s').$ In turn, in the allocation $A[\hat t]$ the $s$-player gets both tasks if he bids $s',$ and he gets only task $2$ with bid $s,$ contradicting WMON for the $s$-player, because $s_1< s_1'$ (cf. Figure~\ref{fig:shapes}).
\end{proof}

Since $f'(.,s_2)$ is increasing as a function of $s_1,$ it must be continuous in all but at most countably many points $s_1.$ Furthermore, since $f'$ is increasing \emph{independently} of $s_2,$ it must be independent of $s_2$ in all $s_1$ where it is continuous, as we show next.

\begin{lemma}\label{lem:indep} If $f'$ is continuous as a function of $s_1$ in some point $(\bar s_1,\bar s_2),$ then it is independent of $s_2,$ i.e., then $f'(\bar s_1,s_2)=f'(\bar s_1, \bar s_2)$ for every $s_2.$ \end{lemma}
\begin{proof} Let $f'(.,\bar s_2)$ be continuous as a function of $s_1$ in the point $\bar s_1.$ Assume for contradiction w.l.o.g., that $f'(\bar s_1,\bar s_2)>f'(\bar s_1,s_2)$ for some $s_2\neq \bar s_2.$ Since $f'(.,\bar s_2)$ is continuous in the point $\bar s_1,$ we can take a $\bar s_1-\delta,$ so that $f'(\bar s_1-\delta,\bar s_2)>f'(\bar s_1,s_2)$ also holds. However, this contradicts Lemma~\ref{lem:cont}, because $\bar s_1-\delta<\bar s_1.$
\end{proof}

We introduce similar notation for other  critical values of the $t$-player based on the virtual payments.  

\begin{notation} If $R_{12}(s)\neq \emptyset$ and $R_{2}\neq \emptyset,$ OR $P_{12}- P_2=0,$ then let $f(s)=P_{12}- P_2.$  Symmetrically, if $R_{12}(s)\neq \emptyset$ and $R_{1}\neq \emptyset,$ OR $P_{12}- P_1=0,$  then let $g(s)=P_{12}- P_1.$  
\end{notation}

As opposed to $f'$ and $g',$ the values $f$ and $g$ are not defined for every $s\in[0,B)\times[0,B);$ this is needed for the following lemmas to hold. Note however that the only case when \emph{neither} $f,$ \emph{nor} $g$ are defined is, when $f'=g'=0,$ and the allocation is quasi-bundling.
The proofs of  Lemmas~\ref{lem:contf} and \ref{lem:indepf} are somewhat more involved variants of those of Lemmas~\ref{lem:cont} and \ref{lem:indep}, and are deferred to the full version.

\begin{lemma}\label{lem:contf} $f$ is monotone increasing in $s_1$ in the following  sense: Let $s_1\leq s_1'$ and  $s_2< s_2'.$  If $f$  is defined for $s=(s_1,s_2),$ then $f$ is also defined for $s'=(s_1',s_2').$ Furthermore $f(s_1,s_2)\leq f(s_1',s_2''),$ whenever $s_1< s_1',$ and $f$ is defined in $(s_1',s_2'').$ Symmetric statements hold for $g.$  \end{lemma}

\begin{lemma}\label{lem:indepf} If $f$ is continuous as a function of $s_1$ in some point $(\bar s_1,\bar s_2),$ then it is independent of $s_2,$ i.e., then $f'(\bar s_1,s_2)=f'(\bar s_1, \bar s_2)$ for every $s_2$ where $f$ is defined, in particular for every $s_2>\bar s_2.$ \end{lemma}

We will sometimes abuse notation and consider $f$ and $f'$ as univariate functions of $s_1.$ This is incorrect only for at most countably many $s_1$ where these functions have jump discontinuities. In these $s_1$ the above single-variate functions remain undefined, or can be identified with, e.g.,  $\lim_{x\rightarrow s_1^-} f(x).$ Analogously we will often treat $g'$ and $g$ as univariate functions.

\bigskip

\noindent\emph{Intuition.} CASE A. below treats all the cases when for some $s$ an allocation function of the $t$-player exists, where at least three (interiors of) regions $R$ are nonempty, the allocation is quasi-flipping or quasi-bundling, and one of the critical value functions $f', g', f, g$ is continuous in $s$ as univariate function. In this case the WMON allocation function over the whole additive domain with bounded $s,$ is a relaxed affine minimizer. The proof of this is based on three observations, roughly summarized as follows: 

First, that once an allocation of the $t$-player is quasi-bundling (quasi-flipping) in some $s$, where at least \emph{one} of $f,g,f',g'$ is continuous as univariate function, we can infinitesimally adjust $s$ in a carefully selected direction, so that the allocation remains quasi-bundling (quasi-flipping), and all of these functions are continuous (when defined) in the adjusted $s.$

Second, that if in any quasi-flipping or quasi-bundling allocation \emph{all} the critical value functions $f',g',f,g$ are continuous in some point $s,$ then they all must have the same derivative in this point. This will follow from the WMON property of the $s$-player. 

Third, that in such points $f'$ and $f$ (if defined) are independent of $s_2,$ and $g',$ and $g$ are independent of $s_1,$ which eventually implies that the above derivatives  must remain the same for \emph{every} $s,$ implying that each of these critical value functions is linear.

In CASE B. a quasi-flipping allocation exists, but $f'\equiv g'\equiv 0$ (for every $s$), and therefore the regions $R_1(s)=R_2(s)=\emptyset$ for every $s.$ In this case the allocation is a 1-dimensional monotone allocation over the sets of positive $s$ and $t,$ where the two tasks are bundled and always allocated to the same player.

Finally, in CASE C. we consider those allocations where the $t$-player has a crossing figure whenever any of $f',g'$ (or  $f,\,g$) is continuous in $s.$ In this case the mechanism is a relaxed task-independent mechanism.

We remark that if in CASE A. the common derivatives of $f,g,f',g'$ happen to be $0,$ or in CASES B. or C. the arbitrary increasing critical value functions turn out to be constant, then we obtain a constant mechanism. We do not discuss these cases separately in the rest of the proof.

\bigskip

We start the argument with three lemmas, stating that the mechanism must be a (relaxed) affine minimizer, if a single point $\hat s$ exists where \emph{all} of $f', g'$ and $f,g$ (if defined) are continuous as univariate functions (i.e., $f'$ and $f$ in the point $\hat s_1$ and $g'$ and $g$ in the point $\hat s_2$), and the allocation $A[\hat s]$ is quasi-flipping or quasi-bundling. 

\begin{lemma}\label{lem:afmin1} Assume that both  $f$ and $g$ are undefined or zero in every $s\in (0,B)\times (0,B).$ If an $\hat s\in(0,B)\times (0,B)$ exists so that the allocation $A[\hat s]$ of the $t$-player is quasi-flipping, $f'>0$ is continuous in the point $\hat s_1,$ and $g'>0$ is continuous in the point $\hat s_2,$ then the mechanism is an affine minimizer (with $\gamma_{12}=\infty$) for nonzero $t$ and $s.$  
\end{lemma}
\begin{proof} The proof makes use of the following claim, which is the core observation implying the linearity of the payment functions whenever an allocation is non-crossing:

\begin{claim}\label{cla:derivative}
Assume that $f'>0$ is continuous as univariate function in the  points $s_1$ and $s_1+\delta,$ and $g'>0$ is continuous in the points $s_2,$ and $s_2+\delta,$ furthermore the allocation of the $t$-player is quasi-flipping in a $2\delta$-neighborhood of  $s=(s_1,s_2).$ Then $f'( s_1+\delta)-f'(s_1)=g'( s_2+\delta)-g'( s_2).$ 
\end{claim}
\begin{proof} (Figure~\ref{fig:linear} (a)) Assume for contradiction w.l.o.g. that $f'(s_1+\delta)-f'( s_1)>g'(s_2+\delta)-g'(s_2).$ Since $f'$ is continuous in $ s_1+\delta,$ it even holds that $f'(s_1+\delta')-f'(s_1)>g'(s_2+\delta)-g'(s_2)$ for some $\delta'<\delta.$

Let $s'=(s_1+\delta',s_2+\delta).$ Now there exists a point $\hat t$ so that $\hat t\in R_2(s)$ but $\hat t \in R_1(s').$ In turn, for this bid $\hat t$ of the $t$-player the $s$-player receives task $1$ if he bids $s,$ and he receives task $2$ if he bids $s'.$ This contradicts WMON for the $s$-player because $s_1'-s_1=\delta'<\delta=s_2'-s_2.$
\end{proof}

\begin{figure}[t]
\centerline{\includegraphics[height=4.5cm]{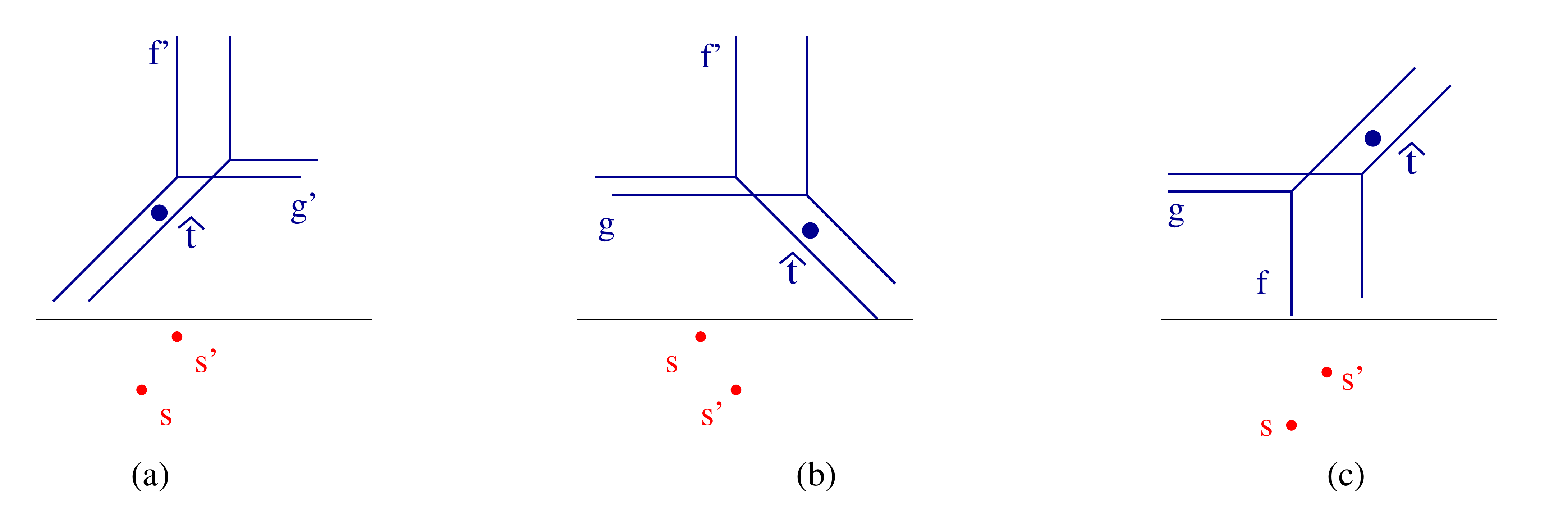}}
\caption{Illustrations to the proofs of Claims~\ref{cla:derivative}, \ref{cla:derivative2}, and \ref{cla:derivative3}}
\label{fig:linear}
\end{figure}

By assumption of the lemma, $f$ or $g$ are never (defined and) positive for any $s\in (0,B)\times (0,B).$ Therefore, the allocation is quasi-flipping for every $s\in(0,B)\times (0,B)$ for which $f'(s)>0$ and $g'(s)>0.$

Let us now fix $\bar s_1,$ and $\delta=1.$ Since $g'$ has a jump in at most countably many points, we can  pick a $\bar s_2$ so that $g'$ is continuous in $\bar s_2+q$ whenever  $q\in \mathbb Q$ and $(\bar s_2+q)\in(0,B).$ Restrict now $g'$ on all such $s_2=\bar s_2+q$ points. For the constant $\Delta=f'(\bar s_1+1)-f'(\bar s_1)$ Claim~\ref{cla:derivative} implies for the restricted $g'$ that $$
g'(s_2+1)=g'(s_2)+\Delta$$ for \emph{every}  $s_2$ in the restricted set. Further, by taking $\delta=1/2,$ we also obtain $
g'(s_2+1/2)=g'(s_2)+\Delta/2,$ and analogously $$
g'(s_2+1/2^r)=g'(s_2)+\Delta/2^r$$ over all $s_2$ in the restriction. By taking the limit $r\rightarrow \infty,$ we obtain that the restriction of $g'$ is linear. Then, using that $g'$ is increasing, and the set of $\bar s_2+q$ points is dense in $(0,B)$ we obtain that $g'$ must be truncated linear over \emph{all} $s_2>0.$  As a further consequence $g'$ is continuous, and independent of $s_1.$ (The short proof presented here does not exclude that  $g'$  jumps to $0$ in a certain $s_2.$ An argument as in the proof of Claim~\ref{cla:derivative} can exclude this.)

By analogous argument we obtain that $f'$ is independent of $s_2$ and linear over all $s_1>0$ for which $f'(s_1)>0.$ Moreover,  by Claim~\ref{cla:derivative} $f'$ has the same multiplicative constant $\mu$ as $g'.$

In summary, for $s>0,$ for the virtual payments of the $t$-player we obtain $P_1(s)=f'(s)=\max(\mu\cdot s_1+\gamma_1,0),$ and $P_2(s)=g'(s)=\max(\mu\cdot s_2+\gamma_2,0),$ and $P_{12}(s)=\max (P_1(s),P_2(s))$ follows from $f=0$ or $g=0.$ 

(We still need to exclude that the allocation is fully bundling for some $s;$ however this is easy to prove, because the region $R_{12}$ could only increase with increasing $s$ and could not disappear in $\bar s$ and in every higher $s.$) The described virtual payments define an affine minimizer over $t\in (0,\infty)\times (0,\infty)$ and $s\in (0,B)\times (0,B).$ If $\mu=0,$ then the mechanism is a constant mechanism (the allocation is independent of $s$). If $t_i=0$ or $s_i=0,$ then the allocation can be different from that of an affine minimizer.
\end{proof}

\begin{lemma}\label{lem:afmin2} Assume that $g'(s_2)=0$  for every $s_2\in [0,B).$ If an $\hat s\in(0,B)\times (0,B)$ exists so that the allocation $A[\hat s]$ of the $t$-player is quasi-bundling, $f'>0$ is continuous in the point $\hat s_1,$ and $g>0$ is continuous in the point $\hat s_2,$ then the mechanism is a relaxed affine minimizer (with $\gamma_2=\infty$).  

For the symmetric case, assume that $f'(s_1)=0$  for every $s_1\in [0,B).$ If an $\hat s\in(0,B)\times (0,B)$ exists so that the allocation $A[\hat s]$ of the $t$-player is quasi-bundling, $g'>0$ is continuous in the point $\hat s_2,$ and $f>0$ is continuous in the point $\hat s_1,$ then the mechanism is a relaxed affine minimizer (with $\gamma_1=\infty$).
\end{lemma}

\begin{proof} Assume w.l.o.g. that the conditions of the first statement of the lemma hold. We use an analogous claim as in the previous lemma:  

\begin{claim}\label{cla:derivative2}
Assume that $f'>0$ is continuous as univariate function in the  points $s_1$ and $s_1+\delta,$ and $g>0$ is continuous in the points $s_2,$ and $s_2-\delta,$ furthermore the allocation of the $t$-player is quasi-bundling in a $2\delta$-neighborhood of  $s=(s_1,s_2).$ then $f'(s_1+\delta)-f'(s_1)=g(s_2)-g(s_2-\delta).$ 
\end{claim}
\begin{proof} (Figure~\ref{fig:linear} (b)) Assume for contradiction first that $f'( s_1+\delta)-f'(s_1)>g( s_2)-g( s_2-\delta).$ Since $f'$ is continuous in $ s_1+\delta,$ it even holds that $f'( s_1+\delta')-f'( s_1)>g( s_2)-g(s_2-\delta)$ for some $\delta'<\delta.$

Let $s'=(s_1+\delta',s_2-\delta).$ Now there exists a point $\hat t$ so that $\hat t\in R_{\emptyset}(s)$ but $\hat t \in R_{12}(s').$ In turn, for this bid $\hat t$ of the $t$-player the $s$-player receives both tasks $1,2$ if he bids $s,$ and he receives no task if he bids $s'.$ This contradicts WMON for the $s$-player because $s_1+s_2>s_1+s_2+\delta'-\delta=s_1'+s_2'.$

Assume now that $f'( s_1+\delta)-f'(s_1)<g( s_2)-g( s_2-\delta).$ Since $f'$ is continuous in $ s_1+\delta,$ it even holds that $f'( s_1+\delta')-f'( s_1)<g( s_2)-g(s_2-\delta)$ for some $\delta'>\delta.$

Let $s'=(s_1+\delta',s_2-\delta).$ Now there exists a point $\hat t$ so that $\hat t\in R_{12}(s)$ but $\hat t \in R_{\emptyset}(s').$ In turn, for this bid $\hat t$ of the $t$-player the $s$-player receives no task if he bids $s,$ and he receives both tasks if he bids $s'.$ This contradicts WMON for the $s$-player because $s_1+s_2<s_1+s_2+\delta'-\delta=s_1'+s_2'.$
\end{proof}

Since $g'\equiv 0,$ the allocation is quasi-bundling whenever $f'(s)>0$ and $g(s)>0.$ By fixing $\bar s_1$ we fix an $f'>0,$ because $f'$ is continuous in $\bar s_1$ and therefore independent of $s_2.$

Like in the proof of Lemma~\ref{lem:afmin1}, we can show that $g(s_2)=P_{12}(s_2)-P_1$ is a (truncated) linear function over $s_2\in (0,B).$ In turn, by fixing $\bar s_2,$ and thus fixing a $g>0,$ we obtain that $f'=P_1(s_1)$ is truncated linear over $s_1\in (0,B)$ with the same derivative as $g.$

Assume finally, that $f'(s_1)=0$ iff $s_1\leq D$, for some $D>0,$ and $g(s_2)= \mu s_2+C$ for an $C>0.$ Then for $(s_1+s_2)\in [0,D)$ the $P_{12}(s_1+s_2)$ is an arbitrary increasing function $P_{12}:[0,D)\rightarrow [0,C),$ because there exists no $s>0$ such that $s_1+s_2<D$ but $f'(s_1)>0$ would hold, and $g$ would be defined (and similarly for $t_1+t_2<C$) so $P_{12}$ needs not be linear over this interval. These $s$ and $t$ values having sum of coordinates below the given constants $D$ and $C,$ respectively, are the only exceptions where the mechanism is a 1-dimensional bundling mechanism and not an affine minimizer.
\end{proof}

\begin{lemma}\label{lem:afmin3} Assume that an $\hat s\in(0,B)\times (0,B)$ exists so that  $f'>0$ and $f>0$ are defined and both are continuous in the point $\hat s_1;$ and $g'>0$ and $g>0$ are defined and both are  continuous in the point $\hat s_2.$ If  the allocation $A[\hat s]$ of the $t$-player is quasi-flipping then the mechanism is an affine minimizer;  if $A[\hat s]$ is quasi-bundling, then the mechanism is a relaxed affine minimizer.  
\end{lemma}
\begin{proof}
Assume first that the allocation $A[\bar s]$ is quasi-flipping, i.e., $0<f(\bar s)<f'(\bar s)$ and $0<g(\bar s)<g'(\bar s).$ A claim analogous to Claim~\ref{cla:derivative} says that $f$ and $g$ must have equal difference quotients:

\begin{claim}\label{cla:derivative3}
Assume that $f>0$ is continuous as univariate function in the  points $s_1$ and $s_1+\delta,$ and $g>0$ is continuous in the points $s_2,$ and $s_2+\delta,$ furthermore the allocation of the $t$-player is quasi-flipping in a $2\delta$-neighborhood of  $s=(s_1,s_2).$ Then $f( s_1+\delta)-f(s_1)=g( s_2+\delta)-g( s_2).$ 
\end{claim}

Let us now fix $s_1=\bar s_1.$ Since $f$ and $f'$ are independent of $s_2$ for this fixed $\bar s_1,$ Claims~\ref{cla:derivative} and \ref{cla:derivative3} can be applied, (as in the proof of Lemma~\ref{lem:afmin1}), to prove that $g$ and $g'$ are truncated linear functions with the same multiplicative constant $\mu$ that is at the same time the derivative of $f$ and $f'$ in the point $\bar s_1.$ (For the applicability of the claims notice that $A[s]$ remains quasi-flipping in some neighborhood of $\bar s,$ since $f<f'$ are continuous in $\bar s_1,$ and in continuous points even independent of $s_2,$ so $f<f'$ must hold in a neighborhood.) Now, fixing $\bar s_2$ and switching the roles of $f$ and $g$ we obtain that $f$ and $f'$ are also truncated linear functions  with multiplicative constant $\mu.$ 

The allocation remains quasi-flipping (because $f<f'$) as long as $f'>0$ and $g'>0.$ At least one of $f$ or $g$ is defined for every $s\in[0,B)\times [0,B).$ Therefore, at least three of $f,g,f',g'$ are always uniquely defined (by the respective truncated linear functions), and they determine the allocation of the $t$-player (up to tie-breaking), which determines the allocation of the $s$-player. The mechanism must be an affine minimizer.

Second, assume that $A[\bar s]$ is quasi-bundling. We can use Claim~\ref{cla:derivative2} to prove that $f,g,f'$ and $g'$ are truncated linear with the same multiplicative constant $\mu$ on their whole domain, and $A[s]$ remains quasi-bundling as long as $f'>0$ or $g'>0$ holds. This determines the allocation as an affine minimizer if $f'>0$ or $g'>0,$ because then at least three of $f,g,f'$ or $g'$ are defined. This is not the case if $f'=g'=0,$ i.e. when the allocation is fully bundling. $P_{12}(x)$ is determined only if an $s=(s_1,s_2)$ with $x=s_1+s_2$ exists so that $f'(s_1)>0$ or $g'(s_2)>0.$ If (for $x<D$ for a given constant $D>0$) no such $s$ exists, then $P_{12}(x)$ is some arbitrary increasing function $P_{12}:[0,D)\rightarrow[0,C)$ and the mechanism becomes 1-dimensional. $C$ is here a constant playing the same role for the $t$ as $D$ for $s.$ All in all, the mechanism is a relaxed affine minimizer.  
\end{proof}

\noindent {\bf CASE A.   $\exists \bar s=(\bar s_1,\bar s_2)\in (0,B)\times (0,B)$ so that $f'(\bar s)>0,$ $f'$ as a function of $s_1$ is continuous in the point $\bar s_1,$ and the allocation of the $t$-player is quasi-flipping or quasi-bundling, (OR the analogous case holds for $g'$);}

\begin{lemma} In CASE A. the conditions of at least one of Lemmas~\ref{lem:afmin1}, \ref{lem:afmin2} or \ref{lem:afmin3} hold, and so the mechanism is a (relaxed) affine minimizer over positive valuations $s$ and $t.$
\end{lemma}
\begin{proof}

\item {\bf Case 1.  in $\bar s$ the allocation of the $t$-player is quasi-flipping}

Since the allocation is quasi-flipping, it follows that $f'>0$ and $g'>0.$ Moreover at least one of $f$ or $g$ is defined  in $\bar s,$ in which case $f<f'$ and $g<g',$ respectively.

\item {\bf Case 1.1     $\,\,f$ is defined for $\bar s$ in the allocation of the $t$-player}

We show first that we can assume w.l.o.g. that all of $f',f$ and $g'$ are continuous as univariate functions in the points $\bar s_1,$ and $\bar s_2,$ respectively, otherwise we could slightly decrease $\bar s_1,$ and after that, slightly increase $\bar s_2,$ so that the continuity  of all these critical value functions holds, and the allocation is still quasi-flipping: 

If $f'$ and $f$ are  both continuous in $\bar s_1,$ then set $\hat s_1=\bar s_1.$ If  $f'$ and $f$ are not \emph{both} continuous in $\bar s_1,$ then we can reduce $\bar s_1$ to $\hat s_1=\bar s_1-\delta$ for an arbitrarily  small value $\delta$ so that $f'$ remains nearly the same and continuous, and  $f$ is continuous (if still defined) in $\hat s_1.$  Since $f'$ does not change much, and $f$ cannot increase, the allocation is still quasi-flipping in $(\hat s_1,\bar s_2).$

We need to show that $f$ is still defined in $(\hat s_1,\bar s_2).$ For this purpose, consider a point $\hat t=(f'(\bar s)-\epsilon, g'(\bar s)-2\epsilon).$ Since the allocation is quasi-flipping in $\bar s,$ it holds that $\hat t\in R_{2}(\bar s),$ but if  $f'(\hat s_1)>\bar s_1-\epsilon,$ then it would get into  $R_{1}$ or into $R_{12},$  unless the allocation is still quasi-flipping with a well-defined $f.$ Therefore, for this $\hat t,$ $\bar s\in R^s_1,$ but $(\hat s_1,\bar s_2)\in R^s_2\cup R^s_{\emptyset}$ would hold, contradicting WMON. 

Now, given that $f'$ and $f$ are continuous in $\bar s_1,$ we can set $\hat s_2\geq \bar s_2,$ so that  $g'$ is continuous, and if $f>0$ then $g$ is (defined and) continuous in $\hat s_2;$   $f$ and $f'$ do  not change, since in $\bar s_1$ both are independent of $s_2.$ Therefore, the allocation in $\hat s$ is quasi-flipping.

\item {\bf Case 1.1.1     $\,\,f(\hat s)>0$}

By Lemma~\ref{lem:afmin3}, in this case the mechanism is an affine minimizer.

\item {\bf Case 1.1.2     $\,\,f(\hat s)=0$}

\item {\bf Case 1.1.2.1  $\,$ both $f$ and $g$ are either zero or undefined for every $(s_1,s_2)$  }

By Lemma~\ref{lem:afmin1}, in this case the mechanism is an affine minimizer.

\item {\bf Case 1.1.2.2     $\,\,\exists s\,$ so that  $f(s)>0$ or $g(s)>0$}

We show that even a  point $s'$ exists where the allocation is quasi-flipping, and each  of $f,f',g,g'$ are positive and continuous in $s'$ in as univariate functions. Then the allocation is an affine minimizer by Lemma~\ref{lem:afmin3}. 

Recall that $f', $ and $f=0$ are continuous in the point $\hat s_1,$ and $g'$ is continuous  in the point $\hat s_2,$ moreover the allocation $A[\hat s]$ is quasi-flipping.
Since by assumption an $s=(s_1,s_2)$ exists so that  $f(s)>0$ or $g(s)>0,$ this must be the case for every $s'$ with  $s'_1\geq s_1$ and  $s'_2\geq s_2,$ by Lemma~\ref{lem:contf}. Therefore we can assume w.l.o.g. that $s_1'> \hat s_1$ and $s_2'>\hat s_2.$
But then, by the same lemma $f$ is defined in $s'$ (since $f$ is defined in $\hat s$), and it cannot be that $f(s')=0,$ because then $g$ would be undefined (or zero) in $s'$ as well. So $f(s')>0$ holds. Since $f(s'),f'(s')$ and $g'(s')$ are all positive, $g(s')$ must be defined and $g(s')>0$ must hold, too. Finally, we assume w.l.o.g. that $f'$ and $f$ are continuous in $s_1'$ and  $g'$ and $g$ are continuous in the point $s_2',$ because the (univariate) continuity holds in almost every $s_1'$ and $s_2'.$

We need to show that the allocation is quasi-flipping in $s'.$ We first consider the allocation in the point $(\hat s_1,s_2').$ Since for $\hat s_1$ the $f$ and $f'$ are independent of $s_2,$ they do not change as $\hat s_2$ gets changed to $s_2'$, and the allocation remains quasi flipping. 
Now let us increase $\hat s_1$ to $s_1'.$ We prove  that  $A[s']$ is still quasi-flipping, so that Lemma~\ref{lem:afmin3} can be applied with $s'$ playing the role of $\hat s.$

Assume for contradiction that $A[s']$ is crossing or quasi-bundling. Since in $s'$ each of $f,f',g,g'$ are defined and continuous, we can decrease $s_2'$ to $s_2'-\delta,$ by an arbitrarily small $\delta,$ the $f$ and $f'$ do not change, and the allocation is still crossing or quasi-bundling. Now there is a point $\hat t=(0,g'(s_2')-\epsilon)$ so that $\hat t\in R_1(\hat s_1,s_2')$ but $\hat t\in R_{12}(s_1',s_2'-\delta).$ This contradicts WMON for the $s$-player for this $\hat t,$ because for $(\hat s_1,s_2')$ he gets task 2, but with bid $(s_1',s_2'-\delta)$ he gets no task.

\item {\bf Case 1.2    $\,\,g$ is defined for $\bar s$ in the allocation of the $t$-player}

There exists a  $s_2^*> \bar s_2$ so that in this point $g'$ and $g$ are both continuous as univariate functions. Note that $f'(\bar s_1,s_2^*)=f'(\bar s)$ by Lemma~\ref{lem:indep}. We claim that in the point $(\bar s_1,s_2^*)$ the allocation of the $t$-player is quasi-flipping, so that  we can use the same argument as in CASE 1.1.1, with $g$ and $g'$ playing the roles of $f$ and $f',$ respectively. This will prove that the mechanism is an affine minimizer.

Assume for contradiction that for $(\bar s_1,s_2^*)$ the allocation is crossing or quasi bundling. Notice that  in this case $g(s_2^*)\geq g'(s_2^*)\geq g'(\bar s_2)>0.$ The first inequality holds since the allocation became crossing or quasi-flipping; the second holds by the monotonicity of $g'.$

Finally, we reduce $\bar s_1$ to $s_1^*=\bar s_1-\delta,$ so that $f'(s_1^*)\approx f'(\bar s_1).$ Let $s^*=(s_1^*,s_2^*).$ Since $g$ and $g'$ are independent of $s_1,$ and $f'(s_1^*)\approx f'(\bar s_1),$ the allocation of the $t$-player is nearly the same for $s^*$ as for $(\bar s_1,s_2^*),$ and in particular $g(s^*)\geq g'(\bar s)$ still holds.

Now, using $g(s^*)\geq g'(\bar s),$  and $f'(s^*)\approx f'(\bar s),$ it can be shown that now there exists a $t$ point so that $t\in R_2(\bar s)$ but $t\in R_{12}(s^{*}).$ Consequently, for this $t$ point, $\bar s\in R^s_{1}(t)$ and $s^{*}\in R^s_{\emptyset}(t),$ contradicting monotonicity, since  $s_1^*=\bar s_1-\delta.$

\item {\bf Case 2.  in $\bar s$ the allocation $A[\bar s]$ of the $t$-player is quasi-bundling } 

Since $f'>0,$ and the allocation is quasi-bundling, $g(\bar s)$ is defined and positive.

\item {\bf Case 2.1    $\,g'(\bar s)>0$}

Since $g'>0,$ and the allocation is quasi-bundling, $f(\bar s)$ is also defined, and $f'<f.$ If $f$ is continuous in the point $\bar s_1$ then set $\hat s_1=\bar s_1.$ Otherwise set $\hat s_1=\bar s_1+\delta$ for some small $\delta>0,$ so that both $f$  and $f'$ are continuous in the point $\hat s_1,$ and therefore independent of $s_2.$ If $\delta$ is small enough then the allocation remains quasi-bundling, because $f$ is increasing in $s_1$ and $f'$ is continuous in the point $\bar s_1.$ Finally, if $g'$ or $g$ are both continuous in the point $\bar s_2,$ then set $\hat s_2=\bar s_2.$ Otherwise set $\hat s_2=\bar s_2+\delta$ so that $g'$ and $g$ are continuous (and positive) in the point $\hat s_2.$ Since $f'<f$ still holds, the allocation is quasi-bundling, and by Lemma~\ref{lem:afmin3} it is a relaxed affine minimizer.

\item {\bf Case 2.2    $\,g'(\bar s)=0$}

\item {\bf Case 2.2.1.  $\,\exists s_2$ so that $g'(s_2)>0$}

Our goal is to increase $\bar s_1$ and $\bar s_2$  so that $g'$ becomes positive, and CASE 1.2.1  applies.

We pick an $\hat t=(f'(\bar s_1)+\epsilon, 0)\in R_{12}(\bar s).$ This is possible, since $A[\bar s]$ is quasi-bundling. It holds then the $s$-player gets no task with bids $\bar s$ and $\hat t.$
We increase $\bar s_1$ to $s_1^*$ so that $f'(s_1^*)<f'(\bar s_1)+\epsilon,$ and  $f'$ is continuous as univariate function in $s_1^*.$ This is possible since $f'$ is continuous in $\bar s_1.$ We also increase $\bar s_2$ to $s_2^*,$ so that $g'$ is strictly positive in the point $s_2^*.$ We claim, that in $s^*=(s_1^*,s_2^*)$ the allocation $A[s^*]$ of the $t$-player is quasi-bundling, and therefore with $s^*$ in the role of $\bar s,$ the proof can be completed as in CASE 1.2.1. 

For $\hat t$ the $s$-player gets no task with bid $\bar s$ (see above), thus by WMON he also gets no task with bid  $s^*,$ since $\bar s<s^*.$ In turn, the $t$-player gets both tasks with bids $s^*$ and $\hat t.$
Given that $f'(s_1^*)<f'(\bar s_1)+\epsilon=\hat t_1,$ this proves that the allocation $A[s^*]$ of the $t$-player is quasi-bundling.

\item {\bf Case 2.2.2.  $\,g'(s_2)=0$ for all $s_2\in [0,B)$}

If $g$ is continuous in the point $\bar s_2,$ then we can apply Lemma~\ref{lem:afmin2} with $\hat s=\bar s.$

If $g$ has a jump in $\bar s_2,$ then let $\hat s_1=\bar s_1,$ and $\hat s_2=\bar s_2+\delta$ so that $g$ is continuous in the point $\hat s_2.$   Then $g$ is defined in $\hat s,$ and  $g(\hat s_2)>g(\bar s_2)>0$ by Lemma~\ref{lem:contf}. Since $g'(\hat s_2)=0$ by the case assumption, $g(\hat s)>g'(\hat s),$ so the allocation is quasi-bundling, and again we can apply Lemma~\ref{lem:afmin2} with $\hat s=\bar s.$

\end{proof}

\noindent{\bf CASE B.  CASE A. does not hold, but $\exists \bar s=(\bar s_1,\bar s_2)\in (0,B)\times (0,B)$ so that $f'(\bar s)=0,$ $f'$ as a function of $s_1$ is continuous in the point $\bar s_1,$ and the allocation of the $t$-player is quasi-bundling, (OR the analogous case holds for $g'$);} 

\begin{claim} In CASE B $f'(s)\equiv 0$ and $g'(s)\equiv 0$ for every $s\in [0,B)\times [0,B).$
\end{claim}

\begin{proof} In $\bar s$ the $f'$ is independent of $s_2,$ whereas $g'$ is increasing in $s_2.$ Assume first for contradiction that $g'>0$ for some $s_2,$ and take an $s_2^*>\bar s_2$ point where $g'$ is continuous, and strictly positive.

Since CASE A does not hold (for $g'$), in $(\bar s_1, s_2^*)$ the allocation of the $t$-player must be crossing, while $f'=0$ still holds.

For $\bar s$ the allocation was bundling, so there is a point $t=(\epsilon, \epsilon)$ that is in $R_{12}(\bar s)$ (because of the bundling alloc.) and in $R_2(\bar s_1+\delta,s_2^*)$ (because of the crossing alloc with $g'>0,$ and because $f'$ is continuous and zero in $\bar s_1,$ it can be continuous and arbitrarily small or zero in some $\bar s_1+\delta$).
For this $t$ in the allocation of the $s$-player $\bar s\in R^s_{\emptyset}$ but $(\bar s_1+\delta,s_2^*)\in R^s_1,$ contradicting WMON. We obtained that $g'\equiv 0.$ By symmetric argument, using the existence of a bundling allocation in $\bar s,$ and continuous $g'$ there, we obtain $f'\equiv 0.$
\end{proof}

\begin{corollary}
Over $t\in (0,\infty)\times (0,\infty)$ and $s\in [0,B)\times [0,B)$ the mechanism is 1-dimensional mechanism bundling the two tasks; for $t_1=0$ and/or $t_2=0$ the allocation of the $s$-player can be non-bundling.
\end{corollary}

\noindent{\bf CASE C.  for every $s\in (0,B)\times (0,B)$ point where either $f'$ or $g'$ is continuous (as univariate functions), the allocation of the $t$-player is crossing}

 It is easy to show that in every point $s$ where $f$ or $g$ as univariate function is continuous in the point $s_1,$ resp $s_2,$ the allocation must be crossing (otherwise we cound change $s$ by a small value so that both $f$ and $f'$ (or $g$ and $g'$) would be continuous and the allocation would be still non-crossing). Thus, the allocation can be quasi-flipping or quasi-bundling, only in points where \emph{each} of $f,g$ (when defined), and $f'$ and $g'$ have a jump discontinuity as univariate functions.
 In this case the allocation function corresponds to a relaxed task-independent mechanism, or 1-dimensional mechanism such  that the $t$-player never gets task 1 or he never gets task 2.

\section{Mechanisms for 2 tasks and 2 machines with submodular valuations}
\label{sec:nonadd}

In this section we relax the assumption of additivity of the valuation of the $t$-player in different ways, and show in  Theorem~\ref{theo:nonaddchar} that for all these extensions of the additive domain the only remaining truthful mechanisms are relaxed affine minimizers (or constant mechanisms), and 1-dimensional mechanisms. This result holds for (i) arbitrary (normalized, monotone) valuations, (ii) submodular or subadditive valuations\footnote{for two tasks subadditive and submodular valuations are the same}, (iii) $\epsilon$-additive valuations (see $V_\epsilon$ below), and (iv) for valuations that are submodular \emph{and} $\epsilon$-additive. In essence, we characterize mechanisms for two tasks and two players in case (i), where the $t$-player has arbitrary monotone valuation $(t_1,t_2,t_{12}),$ and the $s$-player has additive valuation $(s_1,s_2, s_1+s_2)$ with $s_1$ and $s_2$ both bounded by an arbitrarily large fixed value $B.$  We use the characterization for two additive players from Section~\ref{sec:add}, and show that relaxed task-independent mechanisms are not extendable onto the nonadditive domain, unless they are task-independent affine minimizers. Along the proof we argue that our lemmas carry over straightforward to  cases (ii), (iii) and (iv) as well. In particular, task-independent mechanisms cannot be extended  even to an arbitrarily 'narrow' superset $V_\epsilon$ of the (2-dimensional) additive domain in 3-dimensions.

\subsection{Basic concepts and notation}

For simplicity of presentation we consider  any bid of either of the players as a vector of three entries. We summarize the notation for the considered domains of valuations (cost functions) as follows:

\begin{definition} $V\subset \mathbb R_{\geq 0}^3$ is the 3D set of all monotone valuations:
$$V=\{(x_1,x_2,x_{12})\in \mathbb R_{\geq 0}^3\,|\,x_{12}\geq x_1\geq 0,\,x_{12}\geq x_2 \geq 0\};$$
$V_+\subset V $ is the \emph{additive plane}: 
$$V_+=\{(x_1,x_2,x_{12})\in \mathbb R_{\geq 0}^3\,|\, x_1, x_2 \geq 0, x_{12}=x_1+x_2\};$$
$V_{+,B}\subset V_+$ is the additive plane, restricted to $[0,B)\times [0,B)$ for some arbitrarily large fixed $B$  
$$V_{+,B}=\{(x_1,x_2,x_{12})\in \mathbb R_{\geq 0}^3\,\,|\,\,  0\leq x_1,x_2< B, x_{12}=x_1+x_2\};$$
$V_{\epsilon}\subset V$ is the $\epsilon$-neighborhood\footnote{We are quite free to use any reasonable definition of such a neighborhood, e.g., we could have chosen  $x_{12}\in((1-\epsilon)(x_1+x_2), (1+\epsilon)(x_1+x_2)),$ or even the intersection of the latter with the current (additive) neighborhood. The characterization is not very sensitive  to this definition.} of the additive plane for some arbitrary fixed $\epsilon>0:$  
$$V_{\epsilon}=\{(x_1,x_2,x_{12})\in V\,\,|\,\,  x_{12}\in (x_1+x_2-\epsilon, x_1+x_2+\epsilon)\};$$
$V_\mathit{submod}\subset V$ is the 3D set of all submodular valuations; for two tasks this is equivalent with subadditivity:
$$V_\mathit{submod}=\{(x_1,x_2,x_{12})\in V\,|\,\, x_{12}\leq x_1+x_2\}.$$
\end{definition}

Let $V^*\in\{V,V_{\epsilon},V_\mathit{submod},V_\mathit{submod}\cap V_{\epsilon}\},$ and let $(A,\hat P)$ denote the allocation and payment functions of a truthful mechanism on $V^*\times V_{+,B}$. For the 
valuations $t=(t_1, t_2, t_{12})$ and $s=(s_1, s_2, s_{12})$ taken from some given domains $t\in V^*$ and $s\in V_{+,B}$, the allocation is  $A(t,s)=(\alpha_t,\alpha_s),$
where $\alpha_t,\alpha_s\in\{12, 1, 2, \emptyset\}.$ Since both tasks are always allocated, $\alpha_t$ determines $\alpha_s=\bar \alpha_t,$ and for simplicity we often identify $A(t,s)$ with the allocation $\alpha_t$ to the $t$-player.  

For given $s\in V_{+,B},$ and $\alpha\in\{12, 1, 2, \emptyset\},$ the \emph{interior} (w.r.t. $V^*$) of the set of all $t\in V^*$ where the $t$-player is allocated $\alpha$ is denoted by $R_{\alpha}(s)\subset V^*.$   We omit $(s)$ from the notation if $s$ is clear from the context. $R_\alpha(s)$ is an open set in $V^*.$ Similarly, for given $t\in V^*,$  the \emph{interior} (w.r.t. $V_{+,B}$) of the  set of all $s\in V_{+,B}$ where the $s$-player is allocated $\alpha$ is denoted by $R^s_{\alpha}(t)\subset V_{+,B}$.

We will assume regarding allocations of the $t$-player that $R_\emptyset(s)$ is nonempty for every $s.$ This assumption is without loss of generality for allocations $A$ with finite approximation ratio of the makespan.
Since $R_\emptyset\neq \emptyset,$ we can assume w.l.o.g. that the payment function $\hat P$ is normalized, that is, $\hat P_{\emptyset}(s)=0$ for every $s.$

The payments to the $t$-player for taking the tasks $12, 1,$ and $2$ respectively, depending on the bid $s$ are denoted by 
$\hat P_{12}(s),\,\, \hat P_{1}(s),\,\,\hat P_{2}(s).$
We omit the argument $s$ if it is clear from the context.    
The payments to the $s$-player depending on the bid $t$ are denoted by $\hat P^s_{12}(t),\,\,\hat P^s_{1}(t),\,\,\hat P^s_{2}(t).$

Observe that a truthful mechanism $(A,\hat P)$ over $V^*\times V_{+,B}$ is a truthful mechanism when restricted to $V_+\times V_{+,B}.$  We denote this restricted mechanism by $(A|_{V_+\times V_{+,B}},\hat P|_{V_+\times V_{+,B}}).$

We introduce the types of WMON allocations that can occur on $V^*\times V_{+,B}.$

\subsubsection{Relaxed affine minimizers}

\begin{definition} An allocation $A$ is an \emph{affine minimizer}, if there exist positive constants  per player $\mu_t$ and $\mu_s,$ and constants $\gamma_{\alpha} \in \mathbb R\cup\{-\infty,\infty\}$ per allocation (say, here $\alpha=\alpha_t$), so that for every input $(t,s)$ the allocation $A(t,s)$ minimizes over $$\mu_t\cdot t_{12}+\gamma_{12},\qquad \mu_t\cdot t_1+\mu_s\cdot s_2+\gamma_{1},\qquad \mu_s\cdot s_1+\mu_t\cdot t_2+\gamma_{2},\qquad \mu_s(s_1+s_2)+\gamma_{\emptyset}.$$
 \end{definition}

 \begin{definition}
   An allocation $A$ is a \emph{relaxed affine minimizer}, if there
   exist positive constants per player $\mu_t$ and $\mu_s,$ and
   constants $\gamma_{\alpha}$ per allocation $\alpha$ (of the
   $t$-player), furthermore an arbitrary increasing function
   $\xi:[0,\min(\gamma_{1},\gamma_{2})-\gamma_{\emptyset})\rightarrow
   [0,\infty)$ (if the interval
   $[0,\min(\gamma_{1},\gamma_{2})-\gamma_{\emptyset})$ is nonempty)
   with
   $\min(\gamma_{1},\gamma_{2})-\gamma_{12}=\xi(\min(\gamma_{1},\gamma_{2})-\gamma_{\emptyset})$,
   so that for every input $(t,s)$

\begin{itemize}
\item[$a)$] if $\mu_s\cdot(s_1+s_2)\geq \min(\gamma_{1},\gamma_{2})-\gamma_{\emptyset},$ the allocation $A(t,s)$ is that of an affine minimizer with the given constants 
\item[$b)$] if $\mu_s\cdot(s_1+s_2)\leq \min(\gamma_{1},\gamma_{2})-\gamma_{\emptyset},$ then if $\mu_t\cdot t_{12}>\xi(\mu_s(s_1+s_2))$ then the allocation for the $t$-player is $\emptyset$ and if $\mu_t\cdot t_{12}< \xi(\mu_s\cdot(s_1+s_2))$ then it is $12.$
\end{itemize}

\end{definition}

\subsubsection{Constant mechanisms}

In a \emph{constant mechanism} the allocation is independent of the bids of at least one of the players. This property can also be interpreted as being an affine minimizer with multiplicative constant $\mu=0$ for this player. 

\subsubsection{1-dimensional mechanisms}

In a \emph{1-dimensional mechanism} only at most two possible allocations can occur. 
If the two occuring allocations (for the $t$-player) are $\emptyset$ and $12,$  we call the mechanism \emph{bundling mechanism}. 

\begin{definition} 
 In a \emph{bundling mechanism} only  the allocations $\emptyset$ and $12$ can occur. There is an arbitrary, increasing function $\xi:[0,B)\rightarrow [0,\infty)$ so that if $t_{12}>\xi(s_1+s_2)$ then the $t$-player gets $\emptyset,$ and if $t_{12}<\xi(s_1+s_2)$ then the $t$-player gets $12.$ 

If $\xi$ has a jump discontinuity in some point $s_1+s_2$ then the critical value for the $t$-player to get  the tasks may depend on the concrete $(s_1,s_2)$ with the given fixed sum, as long as it is between $\xi((s_1+s_2)^-)$ and $\xi((s_1+s_2)^+).$
\end{definition}

In the other cases the only occuring allocations are $\emptyset$ and $1$ or $\emptyset$ and $2,$ respectively.  The mechanism is in this case identical with a degenerate task-independent mechanism (see also Section~\ref{sec:add}) where task $2$ (resp. task $1$) is always received by the $s$-player. Such mechanisms can also be  extended to $V\times V_{+B}.$  Notice however, that all 1-dimensional mechanisms have  arbitrarily high approximation ratio if we set the fixed bound $B$ high enough.

\subsection{The main result}

\begin{theorem}\label{theo:nonaddchar} Every WMON allocation for two
  tasks and two players with bids $t\in V_{\epsilon}$ and $s\in
  V_{+,B}$, where both tasks are always allocated, and
  $R_\emptyset(s)\neq \emptyset$ for every $s,$ is one of these three
  types: (1) relaxed affine minimizer, (2) one-dimensional mechanism,
  or (3) constant mechanism.  
  The same characterization holds for WMON allocations over $V^*\times
  V_{+,B},$ for
  $V^*\in\{V,V_\mathit{submod},V_\epsilon,V_\mathit{submod}\cap
  V_{\epsilon}\}.$
\end{theorem}

\begin{remark} $(a)$ We note that for $t\in V$ or $t\in V_\mathit{submod},$ Theorem~\ref{theo:nonaddchar} has a shorter, direct proof not using the characterization for additive players of Section~\ref{sec:add}. However, to the best of our knowledge this direct proof does not carry over to the almost-additive domain $V_\epsilon$ for the $t$-player.

\noindent $(b)$ Arguments analogous to the following proof show that the same characterization holds, e.g., for player valuations $(t,s)$  in $ V\times V,$ or $ V_\mathit{submod}\times V_\mathit{submod},$ or $ V\times V_+.$ Since this is not the central topic of this paper, and similar characterizations have been known, we do not consider these cases here: our primary goal is to show that \emph{even} if the $s$-player remains additive and \emph{even} bounded, and the $t$-player is \emph{almost} additive, the 'disturbing' task-independent allocations disappear. 

\end{remark}

\subsection{Proof of Theorem~\ref{theo:nonaddchar}}
\subsubsection{Virtual payments}

Let $(A,\hat P)$  be a truthful mechanism for $(t,s)$ over $V^*\times V_{+,B},$ for $V^*\in\{V,V_\mathit{submod},V_{\epsilon},V_\mathit{submod}\cap V_{\epsilon}\}.$  We will use \emph{virtual payments} for the $t$-player defined as follows.

\begin{definition}\label{def:virt} For every truthful mechanism $(A,\hat P)$ over $(t,s)\in V^*\times V_{+,B},$ or over $(t,s)\in V_+\times V_{+,B},$ we define the \emph{virtual payments} iteratively. Let  $s$ be arbitrary fixed valuation from $V_{+,B}.$ The payments of the $t$-player occuring in this definition are all defined for this given $s.$  

$$P_\emptyset=0,$$
$$P_1=\max\{\hat P_1,0\},\quad P_2=\max\{\hat P_2,0\},$$
$$P_{12}=\max\{\hat P_{12},P_1,P_2\}.$$

\end{definition}

\begin{definition} We call a payment function $P$ \emph{monotone}, if for every $s,$ and sets of tasks $\beta\subset \alpha,$ it holds that  $P_\beta(s)\leq P_\alpha(s).$
\end{definition}

\begin{claim}\label{cla:truthful} The virtual payments are normalized, monotone and  $(A,P)$ is a truthful mechanism.
\end{claim}

\begin{proof} The virtual payment fuction $P$ is obviously normalized and monotone. We show that $P$ is truthful for the $t$-player for arbitrary $s.$ We claim first that $P_\alpha(s)=\hat P_\alpha(s),$ whenever $R_\alpha(s)$ is nonempty. We show this for $\alpha=12,$ the cases $\alpha=1$ and $\alpha=2$ are similar. Assume for contradiction that $t\in R_\alpha(s),$ but  $P_{12}\neq \hat P_{12}.$ In this case either $P_{12}=\hat P_1,$ or $P_{12}=\hat P_2,$ or $P_{12}=0=\hat P_\emptyset,$ so the allocation $\alpha=1,$ or $\alpha=2,$ or $\alpha=\emptyset,$ respectively, is at least as profitable as the more costly allocation $\alpha=12,$ so $t$ must be in the closure of some other region and cannot be in the open interior $R_{12},$ a contradiction. 

In order to prove the truthfulness of $P,$ we need to prove that the given allocation $A(t)$ maximizes the profit $P-t$ for every $t\in V^*$ (resp. for every $t \in V_+$). For every fixed $t$ it holds that $t\in \overline R_\alpha$ for some $\alpha$ with $R_\alpha\neq \emptyset.$ We show next that for this $\alpha,$ $P_\alpha-t_{\alpha}\geq P_\beta-t_\beta$ for all $\beta\in\{12,1,2,\emptyset\}.$

If $R_\beta\neq \emptyset,$ then, by the first paragraph, $P_\alpha=\hat P_\alpha,$ and $P_\beta=\hat P_\beta,$ so the inequality holds, because the $\hat P$ are truthful payments. 
Assume now, that $R_\beta= \emptyset.$ Then $\beta\in\{1,2,12\},$ by the assumption $R_\emptyset\neq\emptyset.$ 
If $\beta=1,$ then $R_1=\emptyset$ and we claim that $P_1=0$ must hold. Otherwise $\hat P_1>0,$ and the points  $t=(t_1<\hat P_1, t_2=L, t_{12}=t_1+L)$ (for large enough $L$) are in $ R_1,$ (by truthfulness of $\hat P$) contradicting $R_1=\emptyset.$ Now, it holds that $P_\alpha-t_\alpha\geq P_\emptyset =0\geq 0-t_1=P_1-t_1.$ Here the first inequality follows from the statement (that $P_\alpha=\hat P_\alpha$ and $P_\emptyset=\hat P_\emptyset$) about nonempty regions, and the truthfulness of $\hat P.$ For $\beta =2$ the proof is analogous. 

Assume finally that $\beta=12.$ We claim that then $P_{12}=P_1$ or $P_{12}=P_2.$ Otherwise $P_{12}=\hat P_{12}>\max \{\hat P_1,\hat P_2,0\},$ and by truthfulness of $\hat P,$ for small enough $\epsilon>0$ the $t=(\epsilon,\epsilon, 2\epsilon)$ must be a point in $R_{12},$ contradicting $R_{12}=\emptyset.$ Assuming now w.l.o.g. $P_{12}=P_1,$ it holds that $P_\alpha-t_\alpha\geq P_1-t_1=P_{12}-t_1\geq P_{12}-t_{12},$ where the first inequality was shown above, and the last inequality holds by monotonicity of the valuations. This concludes the proof.
\end{proof}

In fact, for given WMON allocation function $A(t)$ as function of the bid $t$ of the $t$-player  
(in particular for the allocation function $A[s]$ for arbitrary fixed $s$), 
the virtual payments are the unique payments which are truthful, normalized and monotone, as we show next.

\begin{claim}\label{cla:unique} For every WMON allocation function $A(t)$ of the $t$-player (i.e., for fixed $s$) over $t\in V^*$ or over $t\in V_+,$  $A(t)$ uniquely determines the virtual payments as the only truthful, normalized and monotone  payments. Conversely, the virtual payments  determine the allocation $A(t)$ up to tie-breaking.
\end{claim}
\begin{proof}  For every $V^*$ it is well-known that the payments $P_{12}, P_1, P_2, P_\emptyset$ uniquely determine the truthful allocation up to tie-breaking, since the allocation maximizes $P_\alpha- t_{\alpha}$ over $\alpha\in\{12, 1, 2, \emptyset\}.$ 

We prove the uniqueness of monotone, normalized, truthful payments for a given allocations function $A(t)$.  By definition, $P_\emptyset=0$ must hold. We claim that, if $R_1=\emptyset,$ then $P_1=0$ must be the case. Assume the contrary, then by monotonicity, $P_1>0$ must hold. However, if $P_1>0,$ then the points  $t=(t_1<P_1, t_2=L, t_{12}=t_1+L)\in V_+\subset V^*$ (for large enough $L$) are in $ R_1,$ (by truthfulness) contradicting $R_1=\emptyset.$  If $R_1\neq \emptyset,$ then $P_1=\sup\{t_1\,|\,(t_1,L,t_1+L)\in R_1\}$ is uniquely determined by $A(t)$. (Note that the supremum cannot be infinite, otherwise $R_\emptyset=\emptyset$ would hold.) By symmetric argument, also $P_2$ is uniquely determined by $A(t).$

Consider now $P_{12}.$ We claim that, if $R_{12}=\emptyset,$ then $P_{12}=\max\{P_1,P_2,0\}$ must be the case (these are unique by the paragraph above). Assume the contrary, that $P_{12}>\max\{P_1,P_2,0\}.$ Then, for every small enough $\epsilon>0$ the $t=(\epsilon,\epsilon, 2\epsilon)\in V_+\subset V^*$ must be a point in $R_{12},$ contradicting $R_{12}=\emptyset.$ Finally, if $R_{12}\neq \emptyset,$ then $(0,0,0)\in R_{12}$ and $R_{12}$ must have boundary points $t$ with at least one of $R_{\emptyset},$  or $R_{1}$ or $R_2,$ given that $R_\emptyset\neq \emptyset.$ If $R_{12}$ has a boundary with (say) $R_1,$ then $P_{12}-P_1=\sup\{ t_2 \,|\,(0, t_2, t_2)\in R_{12}   \},$ because this point (i.e., with the supremum $t_2$ value) must be on the boundary with $R_1.$  If $R_{12}$ has no common boundary with $R_1$ or $R_2,$  then $P_{12}=\sup\{ t_2 \,|\,(0, t_2, t_2)\in R_{12}   \},$ because this point  must be on the boundary with $R_\emptyset.$   
\end{proof}

\subsubsection{One-to-one correspondence of allocation and virtual payments}

The virtual payments determine the allocation so that critical values for a bid getting/losing a task (for fixed other bids) are always non-negative.\footnote{In geometric terms, axis-parallel boundaries always have position at least $0.$} Exploiting this, it simplifies the discussion that we can assume that the closure of each of the 4 allocation regions of the $t$-player is always nonempty (over $V^*$ or $V_+):\,\,$ if for some $t\in \mathbb R_{\geq 0}^3$ an allocation (say $12$) does not contradict truthfulness for the virtual payments, then we assume that $t$ is in the closure of the respective allocation region (say $t\in \overline R_{12}). $ For example, w.l.o.g. $(0,0,0)\in \overline R_{12},$  and $(0,L,L)\in \overline R_{1}$ for large enough  $L$  always hold, because $P_{12}$ is the largest payment, and $P_1$ is nonnegative. \footnote{Mind that the closure of a region $\overline R_\alpha$ defined this way is \emph{not} the topological closure of the respective interior $R_\alpha,$ e.g., if $R_\alpha$ is empty.} In the proofs we always pick bids $t$ from the \emph{interior} of some allocation region (here the allocation is uniquely determined), so the above assumption is without loss of generality.

Claim~\ref{cla:unique} has a crucial, and maybe surprising implication concerning a one-to-one relation between  WMON allocations $A(t)$ over $t\in V^*$ and their respective restrictions to the additive plane $A|_{V_+}.$ 

\begin{lemma}\label{lem:same} For every WMON allocation function $A(t)$ over $t\in V^*$ (in particular for the allocation function $A[s]$ for arbitrary fixed $s$), and its restriction to the additive plane $A_+:=A|_{V_+}$ hold that $A$ uniquely determines $A_+,$ and conversely, up to tie-breaking $A_+$ uniquely determines $A$ over $V^*.$ Moreover the virtual payments $P_{12}, P_1, P_2, P_\emptyset$ are the same for $A$ as for $A_+.$
\end{lemma}
\begin{proof} The first statement, that $A$ determines $A|_{V_+},$ is obvious. Since the payments for $A$ are normalized, monotone, and truthful for $A_+,$ by uniqueness these are the only possible  virtual payments for $A_+.$ 

We show how the other implication (that $A_+$ determines $A$) also follows from the uniqueness of the normalized and monotone payments.  Suppose that there exist two different WMON allocation functions over $V^*,$ $A'(t)$ and $A''(t),$ such that $A_+(t)=A'|_{V_+}(t)=A''|_{V_+}(t)$ for every $t\in V_+.$ 
By Claim~\ref{cla:unique} there exist unique virtual payments $P'=(P'_{12},P'_{1},P'_{2},P'_{\emptyset},)$ and $P^2=(P''_{12},P''_{1},P''_{2},P''_{\emptyset},),$ i.e., truthful, normalized and monotone payments for  $A'$ and $A'',$ respectively. Obviously, both are truthful, normalized and monotone for the allocation function $A_+(t),$ when restricted to $V_+.$ Since $A_+$ has unique truthful, normalized and monotone payments, by Claim~\ref{cla:unique}, it must be the case that $P'_{12}=P''_{12},\,P'_{1}=P''_{1},\,P'_{2}=P''_{2},\,$ and $P'_{\emptyset}=P''_{\emptyset}.$ This, in turn, implies by the same claim applied on $V^*$ that $A'\equiv A''$ up to tie-breaking, which concludes the proof. 
\end{proof}

\subsubsection{Illustration} We provide a geometric illustration to  Lemma~\ref{lem:same}. 
For fixed $s,$ the allocation function to the $t$-player $A(t_1,t_2,t_{12})$ has a very simple geometric representation in $\mathbb R_{\geq 0}^3.$ 

Consider the (closures of the) four allocation regions of the $t$-player in $V$: $\, \overline R_{12},\overline R_{1},\overline R_{2},\overline R_{\emptyset}.$ By truthfulness, these regions are uniquely determined by the three payments as follows: 

$$\overline R_{\emptyset}=\{ t\in V \,| \, t_{1}\geq P_1, t_2\geq P_2, t_{12}\geq P_{12}\},$$ 
$$\overline R_{1}=\{ t\in V \,| \, t_{1}\leq P_1, t_1\leq t_{12}+P_1-P_{12}, t_{1}\leq t_2+P_1-P_{2}\},$$
$$\overline R_{2}=\{ t\in V \,| \, t_{2}\leq P_2, t_2\leq t_{12}+P_2-P_{12}, t_{2}\leq t_1+P_2-P_{1}\},$$
$$\overline R_{12}=\{ t\in V \,| \, t_{12}\leq P_{12}, t_{12}\leq t_{1}+P_{12}-P_{1}, t_{12}\leq t_2+P_{12}-P_{2}\}.$$
The boundaries between pairs of these regions, if exist, must be \emph{subsets} of the following planes, respectively:
$$\{ t\in V \,| \, t_1= P_1\}\qquad\hspace{5cm}{(R_1|R_{\emptyset})}$$
$$\{ t\in V \,| \, t_2= P_2\}\qquad\hspace{5cm}{(R_2|R_{\emptyset})}$$
$$\{ t\in V \,| \, t_{12}= P_{12}\}\qquad\hspace{5cm}{(R_{12}|R_{\emptyset})}$$
$$\{ t\in V \,| \, t_{12}-t_1= P_{12}- P_1\}\hspace{5cm}(R_{12}|R_1)$$
$$\{ t\in V \,| \, t_{12}-t_2= P_{12}-P_2\}\hspace{5cm}(R_{12}|R_2)$$
$$\{ t\in V \,| \, t_{1}-t_2= P_{1}- P_2\}\hspace{5cm}(R_{2}|R_1).$$
All these planes have the common intersection point $(P_1,P_2,P_{12}).$ By monotonicity of the virtual payments, this common point is in the domain  $V$ of monotone valuations, and it is a boundary point of all four regions in $V.$ Moreover these boundaries always exist in $V$ (not only their planes). Their positions determine the allocation and the payments and vice-versa. The four regions in $V\subset \mathbb R_{\geq 0}^3,$ of a general WMON allocation  have thus a rather simple and visible structure.

Let us now consider the additive plane $V_+;$ this plane always intersects the first two boundaries ($R_1|R_{\emptyset}$ and $R_2|R_{\emptyset}$), so that the allocation on the additive plane uniquely determines $P_1$ and $P_2.$ Finally, always at least one of the third, fourth and fifth boundaries intersects (as boundary) the additive plane $V_+,$ which determines  $P_{12}.$ 

In particular, as can be seen from the proof of Lemma~\ref{lem:same}, the interior of a region $R_\alpha$ of the $t$-player is empty in $V_+$ if and only if it is empty in $V$ (and consequently, if and only if it is empty in $V^*,$ since $V_+\subset V^*$). This may seem surprising, since it could in principle have been the case that the additive plane $V_+$ does not intersect some nonempty region $R_\alpha\subset V.$

\subsubsection{Extending the additive characterization to $V^*\times V_{+,B}$}

In what follows, we characterize all truthful mechanisms for two tasks and two bidders with  
valuations on the domains  $V\times V_{+,B},$   $V_\mathit{submod}\times V_{+,B},$  $V_\epsilon\times V_{+,B},$ and $(V_\epsilon\cap V_\mathit{submod})\times V_{+,B}$ based on the characterization on $V_+\times V_{+,B}.$ According to Theorem~\ref{theo:addchar}, every truthful mechanism on $V_+\times V_{+,B}$ is one of the following four types: (1) relaxed affine minimizer, (2) relaxed task-independent mechanism,  (3) 1-dimensional mechanism, (4) constant mechanism.

Let $A$ be some WMON  allocation  on $V^*\times V_{+,B}$ for $V^*\in\{V,V_\mathit{submod},V_{\epsilon},V_\mathit{submod}\cap V_{\epsilon}\}.$  (these are the allocations to be characterized). Its restriction $A|_{V_+\times V_{+,B}}$ to additive valuations of the $t$-player is a WMON allocation for additive valuations, and therefore is of one of the types (1) to (4) above. Let $A|_{V_+\times V_{+,B}}$ be for instance a (relaxed) affine minimizer. Then, for every $s\in V_{+,B}$ the allocation and payments for the $t$-player on $t\in V_+$ are those of this affine minimizer. By uniqueness (Lemma~\ref{lem:same}), they extend to $t\in V^*$ uniquely to be the allocation and payments of the same affine minimizer. This, in turn, means that for every fixed $t\in V^*$ the $s$-player on $s\in V_{+,B}$ has allocation and payments according to this affine minimizer. Altogether the mechanism on $V^*\times V_{+,B}$ is the same affine minimizer, with the same parameters, and same virtual payments as on $V_+\times V_{+,B}.$ 
The argument is analogous for 1-dimensional mechanisms and for constant mechanisms. On the other hand, for relaxed task independent mechanisms we show next that these are not extendable to $V^*\times V_{+,B},$ (in particular, not even to $V_\epsilon\times V_{+,B}$(!)), except for those that are at the same time affine minimizers or 1-dimensional mechanisms.


\bigskip

\subsubsection{  $(A,P)|_{V_+\times V_{+,B}}$ is a (relaxed) task-independent mechanism}

A relaxed task-independent mechanism allocates the two tasks by independent 1-dimensional WMON allocations, up to $(s_1,s_2)$ (resp. $(t_1,t_2)$) points where the critical value functions $\phi()$ and $\eta()$ of both of these mechanisms have a jump discontinuity. Since there are at most countably many $s_1, s_2, t_1,$ and $t_2$ where these increasing functions of the 1-dimensional mechanisms can have a jump, it will be easy to handle these $s$ and $t$ values  in the following discussion. 
We assume that the inverse of an increasing function $\phi$ is constant over intervals where the original function has a jump, and conversely, the inverse $\phi^{-1}$ has a jump discontinuity in $c,$ and $\phi^{-1}(c)$ is  defined in an arbitrary monotone way, when the original function is constant $c$ over some interval.

Assume that on $V_+\times V_{+,B}$ the payments to the $t$-player, according to the task-independent mechanism are 
$$P_1(s)=\phi(s_1),\quad P_2(s)=\eta(s_2),\quad P_{12}(s)=\phi(s_1)+\eta(s_2),$$
where $\phi$ and $\eta$ are arbitrary non-negative increasing functions over $s_1,s_2\in [0,B).$ 
We can assume that the values of $\phi$ and $\eta$ are finite, 
because  $R_\emptyset=\emptyset$ would hold for  the $t$-player in every point $(s_1,s_2)$ where $\phi(s_1)=\infty,$ or $\eta(s_2)=\infty.$

By Lemma~\ref{lem:same}, the same payments to the $t$-player  hold over $t\in V^*.$

In what follows, we consider the allocation and payments to the $s$-player on $V^*\times V_{+,B}$ for some carefully chosen $t\in V_\epsilon\cap V_\mathit{submod}\subset V^*$ values, and will conclude that the mechanism is truthful only if $\Phi\equiv 0$ or $\eta\equiv 0$ (implying a 1-dimensional mechanism), or if  $\phi$ and $\eta$ are both linear by the same multiplicative constant, which corresponds to a task-independent affine minimizer.\footnote{A task-independent affine minimizer minimizes among $\,t_{12}, \,\,\, t_1+\mu s_2+\gamma_1, \,\,\, \mu s_1+t_2+\gamma_2,\,\,\, \mu (s_1+s_2)+\gamma_1+\gamma_2$ for some given constants $\mu,$ $\,\gamma_1,\,\gamma_2$.}

We use the notation $\phi(B)=\lim_{s_1\rightarrow B^-} \phi(s_1)$ and  $\eta(B)=\lim_{s_2\rightarrow B^-} \eta(s_2),$
furthermore $\phi(0)=\lim_{s_1\rightarrow 0^+} \phi(s_1)$ and  $\eta(0)=\lim_{s_2\rightarrow 0^+} \eta(s_2).$ 

The case when $\phi(0)=\phi(B),$ \emph{and} $\eta(0)=\eta(B),$ yields a constant mechanism. The cases when $\phi(0)=\phi(B)=0$ or $\eta(0)=\eta(B)=0$ correspond to 1-dimensional mechanisms. The case when $\phi(0)=\phi(B)=C>0$ (or $\eta(0)=\eta(B)=C>0$) on $V_+\times V_{+,B}$ is a task-independent mechanism that allocates one of the tasks by a constant mechanism, independently of the bid $s.$ We sketch a proof that the latter mechanisms do \emph{not} extend to $V^*\times V_{+,B}$ when $C>0.$ Assume the contrary that $\phi\equiv C>0.$ Then for the payments $$P_\emptyset=0,\quad P_1\equiv C, \quad P_2=\eta(s_2),\quad P_{12}=\eta(s_2)+C$$ must hold, even on $V^*\times V_{+,B}.$ This implies that even for every $t\in V^*$ in the extension, the allocation is independent of $s_1.$ That is, for every $t\in V^*$ the allocation to the $s$-player either consists of a single region, or has a single horizontal boundary. However, the above payments to the $t$-player imply that  if $\phi(s_1)\equiv C>0$ and $\eta(s_2)$ takes on at least two different values, then there exists a $t\in V^*$ so that for this fixed $t,$ by increasing $s_2,$ the allocation to the $s$-player gets from $R_2^s(t)$ into $R_1^s(t)$ (and this happens even for every $s_1$). However, the $s$-player cannot have a nonempty  $R_2^s(t)$ and $R_1^s(t)$ regions at the same time, if it only has horizontal boundaries, a contradiction. 

In the rest of this subsection we assume $\phi(0)<\phi(B)$ and $\eta(0)<\eta(B).$

\begin{definition} We will say that 
\begin{itemize}
\item[-] $t_1>0$ is \emph{valid}, if $t_1 \in (\phi(0),\phi(B)),$ and $\phi^{-1}$ is well-defined (continuous) in the point $t_1;$
\item[-] $t_2>0$ is \emph{valid}, if $t_2 \in (\eta(0),\eta(B)),$ and $\eta^{-1}$ is well-defined (continuous) in the point $t_2;$ 
\item[-] $t=(t_1,t_2,t_{12})\in V$ is \emph{valid}, if each of $t_1,$ $t_1'=t_{12}-t_2,$ $\,t_2,$ and $t_2'=t_{12}-t_1$ are valid.
\end{itemize}
\end{definition}

Observe that almost all\footnote{meaning: 'except for a countable subset'} $t_1\in (\phi(0),\phi(B))$ and almost all  $t_2\in (\eta(0),\eta(B))$ are valid, because the increasing functions $\phi^{-1}$ and $\eta^{-1}$ have a jump-discontinuity in at most countably many points. Moreover, if $t_1$ and $t_2$ are valid, then the point $(t_1,t_2,t_1+t_2)\in V_+$ is valid, and for all $\epsilon'<\epsilon,$ the points $(t_1,t_2,t_1+t_2-\epsilon')$ are in $V_\epsilon\cap V_\mathit{submod},$ and they are valid for almost all $\epsilon'.$ It follows that the set of all valid $t$ points is dense in $V_\epsilon\cap V_\mathit{submod},$ at least if we restrict attention  to $t_1 \in (\phi(0)+\epsilon',\phi(B)-\epsilon').$  Our goal is to show that the restriction of $\phi^{-1}$ to the valid points of $(\phi(0)+\epsilon',\phi(B)-\epsilon')$ is identical to a linear function, and  analogously for $t_2$ and $\eta^{-1};$ moreover $\phi^{-1}$ and $\eta^{-1}$ have the same slope. Thus, since a restriction of each of $\phi$ and $\eta$ to a dense subset is linear, and they are increasing, in fact, they cannot have jump discontinuities at all, and are (allover) linear functions with the same slope. Finally, by taking $\epsilon'\rightarrow 0,$ we obtain linearity over $(\phi(0),\phi(B)),$ and $(\eta(0),\eta(B))$ respectively. 
 This will prove that $A|_{V_+\times V_{+,B}}$ is a task-independent affine minimizer, and  $A$ is necessarily the same affine minimizer on $V^*\times V_{+,B}.$

First we argue that for all valid $t$ the regions $R^s_\emptyset(t)\subset V_+$ and $R^s_1(t)\subset V_+$ are nonempty, in particular there exist $s\in V_{+,B}$ on the boundary of these two regions. Since $R^s_\emptyset(t)$ is nonempty, for valid $t$ the virtual payments $P_\alpha^s(t)$ of the $s$-player are well-defined. Further, an $s$ on the boundary will help us to obtain a formula for $P_1^s(t).$

Consider an arbitrary fixed valid $t\in V.$ Since $t_2< \eta(B),$ there exist $s\in V_{+,B}$ so that $t_2<\eta(s_2),$ or equivalently $\eta^{-1}(t_2)< s_2.$ For every such $s$ and the fixed $t,$ the $t$-player receives task $2.$ Moreover, since $\eta^{-1}$ is continuous in $t_2,$ the same holds (there exist proper $s_2$ values globally) for \emph{all} points in some neighborhood of $t.$ 
Now, for every such $s,$ the $s$ is on the boundary of $R^s_{\emptyset}$ and $R^s_1,$ iff our fixed point $t$ is on the boundary of $R_{12}(s)$ and $R_2(s).$\footnote{In fact, the boundary of the $t$-player is allowed to have a jump in $s_1=\phi^{-1}(t_1).$} On the other hand, $s$ is on the boundary of $R^s_{\emptyset}$ and $R^s_1,$ iff $s_1=P^s_1$ (here we use that $s$ can not get task $2$), and $t$ is on the boundary of $R_{12}(s)$ and $R_2(s)$ iff $t_{12}-t_2=P_{12}(s)-P_2(s)=\phi(s_1).$ Observe that such an $s_1$ does exist, given that $t_{12}-t_2$ is valid, i.e. is within the range of valid $t_1$ values and is a continuity point of  $\phi^{-1}.$ In summary, for the fixed valid $t,$ and the found $s$ it must hold that $$P^s_1(t)=s_1=\phi^{-1}(t_{12}-t_2),$$ which, for any valid $t$ yields the formula for $P^s_1(t)$ independently of $s.$ Analogously, we obtain for valid $t$ that $$P^s_2(t)=\eta^{-1}(t_{12}-t_1).$$

Next we determine $P_{12}^s$ in two different ways for every valid $t.$ First of all we claim that if $t$ is valid then $R_{12}^s(t)$ is nonempty.
Indeed, by the validity of $t,$ there exist $s$ such that $s_1<\phi^{-1}(t_1),$ and $s_2<\eta^{-1}(t_2),$ and $s_2<\eta^{-1}(t_{12}-t_1).$ For all such  $s$ it holds that $t_1>\phi(s_1)=P_1(s),$ and $t_2>\eta(s_2)=P_2(s),$ and $t_{12}>\eta(s_2)+t_1>\eta(s_2)+\phi(s_1)=P_{12}(s).$ Therefore $t\in R_\emptyset(s).$ This implies that the (open) set of all such $s$ must belong to $R_{12}(t).$ 

Second, we claim that there exists an $s'$ on the boundary of $R_{12}^s(t)$ and $R_2^s(t),$ and another $s''$ on the boundary of $R_{12}^s(t)$ and $R_1^s(t).$
Indeed, let $s'$ be such that $s'_1=\phi^{-1}(t_1),$ and $s'_2<\eta^{-1}(t_2),$ and $s'_2<\eta^{-1}(t_{12}-t_1).$ Then, $t_2>\eta(s'_2)=P_2(s'),$ and $t_{12}>\eta(s'_2)+t_1=\eta(s'_2)+\phi(s'_1)=P_{12}(s'),$ so the $t$-player does not get task $2$ even with a bid in a small enough neighborhood of $t,$ and because of $t_1=\phi(s'_1)$ it must be on the boundary of $R_1(s')$ and $R_\emptyset(s'),$ so $s'$ is on the boundary of $R^s_2(t)$ and $R^s_{12}(t)$ (here we use that $s'_1$ is uniquely defined as $\phi^{-1}(t_1)$). The proof is analogous for $s''.$

By the position of $s',$
$$P_{12}^s-P_2^s=s'_1=\phi^{-1}(t_1);$$
the second equality follows from the fact that for $s'$ the $t$ must be on the boundary of $R_\emptyset$ and $R_1.$
Plugging in $P^s_2$ we obtain (using also that $\phi^{-1}$ and $\eta^{-1}$ are well-defined in $t$)
$$P_{12}^s=\phi^{-1}(t_1)+\eta^{-1}(t_{12}-t_1).$$
On the other hand, using $s''$ we obtain by analogous argument that
$$P_{12}^s=\eta^{-1}(t_2)+\phi^{-1}(t_{12}-t_2).$$
This yields that $$\phi^{-1}(t_1)+\eta^{-1}(t_{12}-t_1)=\eta^{-1}(t_2)+\phi^{-1}(t_{12}-t_2)$$
for \emph{every} valid $t.$

Let $\epsilon'<\epsilon.$ Assume for a moment that \emph{all} $t\in V_\epsilon\cap V_\mathit{submod}$ with $t_1 \in (\phi(0)+\epsilon',\phi(B)-\epsilon')$ and $t_2 \in (\eta(0)+\epsilon',\eta(B)-\epsilon')$ are valid, and consider such a $t.$ Increasing now only $t_1$ by an arbitrary positive $\delta<\epsilon',$ the right side of the last equality does not change. This is possible only if $\phi^{-1}(t_1)$  and $\eta^{-1}(t_{12}-t_1)$ increases, resp. decreases by the same value $\Delta:=\phi^{-1}(t_1+\delta)-\phi^{-1}(t_1)=\eta^{-1}(t_{12}-t_1)-\eta^{-1}(t_{12}-t_1-\delta).$  

Now fix $t_1, t_2$ and $\delta$ (thereby fixing $\Delta$),  and consider \emph{every} $t_{12}\in (t_1+t_2-\epsilon',t_1+t_2).$ For every such $t_{12}$ holds by the same argument  $\eta^{-1}(t_{12}-t_1)-\eta^{-1}(t_{12}-t_1-\delta)=\Delta$ with the fixed $\Delta.$ The last statement can be rewritten as $\eta^{-1}(x)-\eta^{-1}(x-\delta)=\Delta\,$  for every $\,x\in (t_2-\epsilon',t_2).$ Now taking $\delta/2^r$ instead of $\delta$ and letting $r\rightarrow \infty,$ by the monotonicity of $\eta^{-1}$ this implies that $\eta^{-1}$ must be \emph{linear} on the interval $(t_2-\epsilon',t_2).$  (The proof of this  is analogous to the proof  of Lemma~\ref{lem:afmin1}.)

As a final step, change $t_2$ to an arbitrary other valid $t_2'.$ Here again, by the same argument $\eta^{-1}$ is linear over $(t_2'-\epsilon', t_2').$
Since such overlapping open intervals cover the interval $(\eta(0)+\epsilon',\eta(B)-\epsilon'),$ the function $\eta^{-1}$ must be a single linear function over $(\eta(0)+\epsilon',\eta(B)-\epsilon').$
By taking the limit $\epsilon'\rightarrow 0$ we obtain that $\phi^{-1}$ is linear over $(\eta(0),\eta(B)).$
By symmetric argument, $\phi^{-1}$ is linear, too, and $\phi^{-1}$ and $\eta^{-1}$ have the same multiplicative constant $\Delta/\delta$ as follows from $\Delta=\phi^{-1}(t_1+\delta)-\phi^{-1}(t_1)=\eta^{-1}(t_{12}-t_1)-\eta^{-1}(t_{12}-t_1-\delta).$   

Now, how do we go around non-valid points $t\in V_\epsilon\cap V_\mathit{submod}?$ By counting arguments it can be shown that the grid of rational $(t_1,t_2)$ pairs can be translated by a real vector, so that all translated pairs are valid. After this, a real $\delta$ can be found so that every corresponding $(t_1,t_2, t_1+t_2-\delta\cdot k/2^r)$ is valid, too. Using only these numbers we obtain that $\phi$ and $\eta$ are linear on the translated rational grid, so given that they are monotone increasing, they must be linear functions allover.
This concludes the proof that the allocation function is a task-independent affine minimizer on $V^*\times V_{+,B}.$

\subsubsection{Mechanisms with high approximation ratio}

We proved above that any truthful mechanism $(A,P)$ on $V^*\times V_{+,B}$ is (1) relaxed affine minimizer, (2) a one-dimensional mechanism, or  (3) a constant mechanism. It remains to prove that in cases (2), and (3), the approximation ratio is high. This will conclude the proof of Theorem~\ref{theo:nonaddchar}.

\begin{lemma} If $(A,P)$ on $V^*\times V_{+,B}$ is a one-dimensional
  mechanism, or a constant mechanism, there exists an instance for
  which the makespan of $A$ is at least $\sqrt{B}$ and the approximation
  ratio is at least $\sqrt{B}$.
\end{lemma}
\label{lem:highapprox}

\begin{proof} Assume first that $A$ is a constant mechanism. By the
  assumption that $R_\emptyset(s)\neq \emptyset$ for every $s,$ it
  must be the case that the constant mechanism is independent of $s.$

  If, for $t=(\sqrt{B},\sqrt{B})$ at least one of the tasks is given
  to the $t$-player, then the makespan is $2\sqrt{B}$ and the
  approximation ratio is unbounded, because $s$ can be chosen
  arbitrarily small. If for $t=(\sqrt{B},\sqrt{B})$ both tasks are
  given to the $s$-player, then we can choose $s=(B,B)$ to prove
  approximation ratio of at least $\sqrt{B}.$

  Assume now that $A$ is a bundling allocation, and consider the input
  $(t_1=\sqrt{B}, t_2=0, s_1=0, s_2=\sqrt{B}).$ Either both tasks are
  given to the $t$-player, or both are given to the $s$-player, so the
  makespan is $\sqrt{B}$ and the approximation ratio is unbounded.

  Finally, assume w.l.o.g., that $A$ realizes only the two allocations
  $\alpha_t=\emptyset,$ or $\alpha_t=1.$ Then set $s_2=\sqrt{B},$ and
  $t_1=t_2=0.$ Since the $s$-player must receive task $2,$ the
  approximation is again unbounded.
\end{proof}


\begin{thebibliography}{10}

\bibitem{AK08}
Aaron Archer and Robert Kleinberg.
\newblock Truthful germs are contagious: A local to global characterization of
  truthfulness.
\newblock In {\em ACM Conference on Electronic Commerce (EC)}, 2008.

\bibitem{AT01}
Aaron Archer and {\'E}va Tardos.
\newblock Truthful mechanisms for one-parameter agents.
\newblock In {\em Proc. of the 42nd IEEE Symposium on Foundations of Computer
  Science (FOCS)}, pages 482--491, 2001.

\bibitem{ADL09}
Itai Ashlagi, Shahar Dobzinski, and Ron Lavi.
\newblock Optimal lower bounds for anonymous scheduling mechanisms.
\newblock {\em Mathematics of Operations Research}, 37(2):244--258, 2012.

\bibitem{Auletta0P15}
Vincenzo Auletta, George Christodoulou, and Paolo Penna.
\newblock Mechanisms for scheduling with single-bit private values.
\newblock {\em Theory Comput. Syst.}, 57(3):523--548, 2015.

\bibitem{BCR+06}
S.~Bikhchandani, S.~Chatterji, R.~Lavi, A.~Mu'alem, N.~Nisan, and A.~Sen.
\newblock Weak monotonicity characterizes deterministic dominant strategy
  implementation.
\newblock {\em Econometrica}, 74(4):1109--1132, 2006.

\bibitem{ChawlaHMS13}
Shuchi Chawla, Jason~D. Hartline, David~L. Malec, and Balasubramanian Sivan.
\newblock Prior-independent mechanisms for scheduling.
\newblock In Dan Boneh, Tim Roughgarden, and Joan Feigenbaum, editors, {\em
  Symposium on Theory of Computing Conference, STOC'13, Palo Alto, CA, USA,
  June 1-4, 2013}, pages 51--60. {ACM}, 2013.

\bibitem{ChenDZ15}
Xujin Chen, Donglei Du, and Luis~Fernando Zuluaga.
\newblock Copula-based randomized mechanisms for truthful scheduling on two
  unrelated machines.
\newblock {\em Theory Comput. Syst.}, 57(3):753--781, 2015.

\bibitem{CKK10}
George Christodoulou, Elias Koutsoupias, and Annam{\'a}ria Kov{\'a}cs.
\newblock Mechanism design for fractional scheduling on unrelated machines.
\newblock {\em ACM Transactions on Algorithms}, 6(2), 2010.

\bibitem{ChristodoulouKV08}
George Christodoulou, Elias Koutsoupias, and Angelina Vidali.
\newblock A characterization of 2-player mechanisms for scheduling.
\newblock In {\em Algorithms - {ESA} 2008, 16th Annual European Symposium,
  Karlsruhe, Germany, September 15-17, 2008. Proceedings}, pages 297--307,
  2008.

\bibitem{ChrKouVid09}
George Christodoulou, Elias Koutsoupias, and Angelina Vidali.
\newblock A lower bound for scheduling mechanisms.
\newblock {\em Algorithmica}, 55(4):729--740, 2009.

\bibitem{CK13}
George Christodoulou and Annam{\'a}ria Kov{\'a}cs.
\newblock A deterministic truthful ptas for scheduling related machines.
\newblock {\em SIAM J. Comput.}, 42(4):1572--1595, 2013.

\bibitem{Cla71}
Edward~H. Clarke.
\newblock Multipart pricing of public goods.
\newblock {\em Public Choice}, 8, 1971.

\bibitem{DaskalakisW15}
Constantinos Daskalakis and S.~Matthew Weinberg.
\newblock Bayesian truthful mechanisms for job scheduling from bi-criterion
  approximation algorithms.
\newblock In Piotr Indyk, editor, {\em SODA}, pages 1934--1952. SIAM, 2015.

\bibitem{DDDR11}
Peerapong Dhangwatnotai, Shahar Dobzinski, Shaddin Dughmi, and Tim Roughgarden.
\newblock Truthful approximation schemes for single-parameter agents.
\newblock {\em SIAM J. on Computing}, 40(3):915--933, 2011.

\bibitem{Dobzinski16}
Shahar Dobzinski.
\newblock Breaking the logarithmic barrier for truthful combinatorial auctions
  with submodular bidders.
\newblock In {\em Proceedings of the 48th Annual {ACM} {SIGACT} Symposium on
  Theory of Computing, {STOC} 2016, Cambridge, MA, USA, June 18-21, 2016},
  pages 940--948, 2016.

\bibitem{DobzinskiN15}
Shahar Dobzinski and Noam Nisan.
\newblock Multi-unit auctions: Beyond roberts.
\newblock {\em J. Economic Theory}, 156:14--44, 2015.

\bibitem{DS08}
Shahar Dobzinski and Mukund Sundararajan.
\newblock On characterizations of truthful mechanisms for combinatorial
  auctions and scheduling.
\newblock In {\em Proceedings 9th {ACM} Conference on Electronic Commerce
  (EC-2008), Chicago, IL, USA, June 8-12, 2008}, pages 38--47, 2008.

\bibitem{DobzinskiV16}
Shahar Dobzinski and Jan Vondr{\'{a}}k.
\newblock Impossibility results for truthful combinatorial auctions with
  submodular valuations.
\newblock {\em J. {ACM}}, 63(1):5:1--5:19, 2016.

\bibitem{EpsteinLS13}
Leah Epstein, Asaf Levin, and Rob van Stee.
\newblock A unified approach to truthful scheduling on related machines.
\newblock In {\em Proc. of the Twenty-Fourth Annual ACM-SIAM Symposium on
  Discrete Algorithm (SODA)}, pages 1243--1252, 2013.

\bibitem{GiannakopoulosK17}
Yiannis Giannakopoulos and Maria Kyropoulou.
\newblock The {VCG} mechanism for bayesian scheduling.
\newblock {\em {ACM} Trans. Economics and Comput.}, 5(4):19:1--19:16, 2017.

\bibitem{Gro73}
Theodore Groves.
\newblock Incentives in teams.
\newblock {\em Econometrica}, 41(4):617--631, 1973.

\bibitem{Rob79}
Roberts Kevin.
\newblock The characterization of implementable choice rules.
\newblock {\em Aggregation and Revelation of Preferences}, pages 321--348,
  1979.

\bibitem{KV07}
Elias Koutsoupias and Angelina Vidali.
\newblock A lower bound of 1+{$\phi$} for truthful scheduling mechanisms.
\newblock {\em Algorithmica}, pages 1--13, 2012.

\bibitem{LaviS09}
Ron Lavi and Chaitanya Swamy.
\newblock Truthful mechanism design for multidimensional scheduling via cycle
  monotonicity.
\newblock {\em Games and Economic Behavior}, 67(1):99--124, 2009.

\bibitem{LehmannLN06}
Benny Lehmann, Daniel~J. Lehmann, and Noam Nisan.
\newblock Combinatorial auctions with decreasing marginal utilities.
\newblock {\em Games and Economic Behavior}, 55(2):270--296, 2006.

\bibitem{LeucciMP18}
Stefano Leucci, Akaki Mamageishvili, and Paolo Penna.
\newblock No truthful mechanism can be better than \emph{n} approximate for two
  natural problems.
\newblock {\em Games and Economic Behavior}, 111:64--74, 2018.

\bibitem{Lu09}
Pinyan Lu.
\newblock Internet and network economics, 5th international workshop, wine
  2009, rome, italy, december 14-18, 2009. proceedings.
\newblock In {\em Proc. of the 5th International Workshop on Internet and
  Network Economics (WINE)}, volume 5929 of {\em Lecture Notes in Computer
  Science}, pages 30--41. Springer, 2009.

\bibitem{LuYu08}
Pinyan Lu and Changyuan Yu.
\newblock An improved randomized truthful mechanism for scheduling unrelated
  machines.
\newblock In {\em 25th Annual Symposium on Theoretical Aspects of Computer
  Science (STACS)}, volume~1 of {\em LIPIcs}, pages 527--538, 2008.

\bibitem{LuY08a}
Pinyan Lu and Changyuan Yu.
\newblock Randomized truthful mechanisms for scheduling unrelated machines.
\newblock In {\em 4th International Workshop on Internet and Network Economics
  (WINE)}, pages 402--413, 2008.

\bibitem{MualemS18}
Ahuva Mu'alem and Michael Schapira.
\newblock Setting lower bounds on truthfulness.
\newblock {\em Games and Economic Behavior}, 110:174--193, 2018.

\bibitem{NR01}
Noam Nisan and Amir Ronen.
\newblock Algorithmic mechanism design.
\newblock {\em Games and Economic Behavior}, 35:166--196, 2001.

\bibitem{SY05}
Michael~E. Saks and Lan Yu.
\newblock Weak monotonicity suffices for truthfulness on convex domains.
\newblock In {\em Proceedings 6th ACM Conference on Electronic Commerce (EC)},
  pages 286--293, 2005.

\bibitem{Vic61}
William Vickrey.
\newblock Counterspeculations, auctions and competitive sealed tenders.
\newblock {\em Journal of Finance}, 16:8--37, 1961.

\bibitem{Yu09}
Changyuan Yu.
\newblock Truthful mechanisms for two-range-values variant of unrelated
  scheduling.
\newblock {\em Theoretical Computer Science}, 410(21-23):2196--2206, 2009.

\end{thebibliography}
\end{document}